\documentclass[12pt,halfline,a4paper]{ouparticle}

\usepackage{float}
\usepackage{amsthm}
\usepackage{xcolor}
\usepackage{dsfont}
\usepackage{graphicx}
\usepackage{subcaption}
\usepackage{hyperref}
\usepackage{booktabs}
\usepackage{tabularx}
\usepackage{array}
\usepackage{makecell}
\usepackage{enumitem}

\newtheoremstyle{boldhead}{3pt}{3pt}
{\itshape}{}{\bfseries}{.}{ }
{\thmname{#1}\thmnumber{ #2}\thmnote{ [#3]}}
\theoremstyle{boldhead}

\newtheorem{proposition}{Proposition}

\newtheorem{cor}{Corollary}

\newtheorem{ex}{Example}

\usepackage{mathtools}

\begin{document}

\title{Generalizing Markowitz Portfolio Optimization by a Quadratic Risk Measure}

\author{%
\name{Ignas Gasparavičius, Andrius Grigutis}
\address{Institute of Mathematics, Vilnius University\\
Naugarduko g. 24, LT-03225, Vilnius, Lithuania}
\email{ignas.gasparavicius@mif.stud.vu.lt, andrius.grigutis@mif.vu.lt}
}

\abstract{We show that the key optimization results of the classical Markowitz portfolio selection theory, originally formulated for variance as the risk measure, remain available in explicit closed form under a broader class of strictly convex quadratic risk measures. The proposed framework replaces the covariance matrix with an arbitrary symmetric positive definite matrix and allows additional linear and constant terms, thereby containing
various models arising in transaction cost optimization, benchmark relative optimization, covariance regularization, and factor models. Closed-form formulas are obtained for the efficient frontier, the global minimum risk portfolio, the maximum Sharpe ratio portfolio,
the Capital Market Curve, the tangency portfolio, and the maximum utility portfolio. In contrast to the classical Markowitz model, the tangency portfolio does not coincide with the maximum Sharpe ratio portfolio, revealing a new geometric phenomenon. A numerical example confirms the derived formulas.}

\date{\today}



\maketitle

\section{Introduction}\label{sec1}

\subsection{Prior work, notations}
This article is motivated by the work \cite{GG2024}, which was devoted to the mathematical foundations of Markowitz's portfolio selection theory, also known as modern portfolio theory. In modern portfolio theory, the risk measure is variance. In fact, variance is one of the simplest risk measures and admits explicit closed-form formulas for the efficient frontier, involved portfolios, and other quantities. On the other hand, variance has limitations in properly estimating tail risk, considering gains and losses equally, etc. This can be improved by using other risk measures such as conditional value at risk, expected shortfall, copula-based measures, and other downside risk-sensitive measures.  However, for many alternative risk measures, the analytical tractability is lost, and the efficient frontier and related portfolios must be computed numerically. 

In this work, we introduce a generalized quadratic risk measure $R_P$ \eqref{mean_var_def} for which the entire portfolio optimization theory remains analytically tractable. The introduced risk measure $R_P$ gives considerably more flexibility than variance while remaining available in explicit closed form, and reveals geometric phenomena that disappear in the classical mean-variance model: the tangency portfolio is not the same as the maximum Sharpe ratio portfolio; see the comment beneath Proposition \ref{P6}. To the best of our knowledge, explicit closed-form expressions of this generality are not available in the existing literature for quadratic risk measures beyond variance. 

Let us denote the set of portfolios
\begin{align*}
\mathcal{P}=\left\{{\pmb w}:=(w_1,\,w_2,\,\ldots,\,w_n)\in\mathbb{R}^n: \, \sum_{i=1}^{n}w_i=1\right\},\quad n\in\mathbb{N}.
\end{align*}

Thus, a portfolio is defined as a vector $\mathbf {\pmb w}\in\mathcal P$, where each component (weight) $w_i$ denotes the fraction of the total investment allocated to the $i$'th asset $A_i$, $i=1,\,2,\,\ldots,\,n$. If $w_i<0$ for some $i$, then asset $i$ is said to be held in a short position; that is, the asset is borrowed and sold with the intention that it can later be repurchased at a lower price.

Let ${\pmb R}:=(r_1,\,\ldots,r_n)$ denote a collection of real-valued random variables defined on the same probability space, representing the random returns of assets $A_1,\,\ldots,\,A_n$. Let
\begin{align*}
{\pmb \mu}:=(\mu_1,\,\mu_2,\,\ldots,\,\mu_n)=\left(\mathbb{E}r_1,\,\mathbb{E}r_2,\,\ldots,\,\mathbb{E}r_n\right)
\end{align*}
be the vector of expected returns, and
\begin{align*}
Q:= \begin{pmatrix}
q_{1,1} & q_{1,2}& \ldots & q_{1,n}\\
q_{2,1}&q_{2,2}&\ldots&q_{2,n}\\
\vdots & \vdots & \ddots & \vdots\\
q_{n,1} & q_{n,2}& \ldots & q_{n,n}
\end{pmatrix}
\end{align*}
be a symmetric positive definite matrix characterizing the dependence structure of the asset returns $r_1$, $r_2$, $\ldots$, $r_n$. For instance, $Q$ can be based on a covariance matrix.

Let the random variable $P:=w_1r_1+w_2r_2+\cdots+w_nr_n={\pmb w}{\pmb R}^T$ denote the random return of the portfolio ${\pmb w}$. For the portfolio's return $P$, we define its expected value $\mu_P$ and the generalized quadratic risk measure $R_P$:
\begin{align}\label{mean_var_def}
\mu_P({\pmb w}):={\pmb \mu}{\pmb w}^T,\quad R_P({\pmb w}):=\frac{1}{2}{\pmb w}Q{\pmb w}^T + {\pmb c}{\pmb w}^T + k.
\end{align}
The vector ${\pmb c}:=(c_1,\,c_2,\,\ldots,\,c_n)\in\mathbb{R}^n$ is referred to as the linear adjustment vector. It models investor-specific preferences and may represent transaction costs, benchmark effects, penalties, or other linear preference terms. The constant $k\in\mathbb{R}$ serves as a normalization constant. It is defined as a real number such that
\begin{align}\label{k_def}
k\geqslant \frac{1}{2}{\pmb c}Q^{-1}{\pmb c}^T.
\end{align}
This guarantees $R_P({\pmb w})\geqslant0$ for all ${\pmb w}\in\mathcal{P}$. Indeed, by completing the square,
\begin{align}\label{risk_non_neg}
R_P({\pmb w})=\frac{1}{2}{\pmb w}Q{\pmb w}^T+{\pmb c}{\pmb w}^T + k = \frac{1}{2}\left({\pmb w} +{\pmb c}Q^{-1}\right)Q\left({\pmb w} +{\pmb c}Q^{-1}\right)^T - \frac{1}{2}{\pmb c}Q^{-1}{\pmb c}^T + k\geqslant0
\end{align}
iff \eqref{k_def} holds. It is clear that the risk measure $R_P$ reduces to variance if ${\pmb c}={\pmb 0}$, $k=0$, and $Q=2\Sigma$, where $\Sigma$ is the covariance matrix. 

In the classical mean-variance framework, the investor maximizes expected return penalized by variance. Replacing the covariance matrix by a general symmetric positive definite matrix $Q$, and allowing a linear penalty ${\pmb c}{\pmb w}^T$ together with a constant $k$, leads to a more general quadratic utility function. This retains the Markowitz trade-off between reward and risk, but allows the risk criterion to include effects not captured by covariance alone. 

The purpose of this paper is to show that, despite this additional generality, the principal results of Markowitz portfolio theory, including the minimum risk portfolio, the efficient frontier, the maximum Sharpe ratio portfolio, the mutual fund separation theorem, and the maximum utility portfolio, remain available in explicit closed-form.

\subsection{Basic properties of risk measures}

In this subsection, we summarize the fundamental properties of the generalized quadratic risk measure $R_P$ defined in \eqref{mean_var_def} and explain how these properties influence the geometry of portfolio optimization in the risk--return plane. According to \eqref{risk_non_neg}, $R_P(\pmb{w})\geqslant 0$ for all ${\pmb w} \in \mathcal{P}$. Moreover, it is a strictly convex risk measure, since its Hessian matrix is $Q$, i.e.,
\begin{align*}
\nabla^2R_P({\pmb w}) = Q,
\end{align*}
where $Q$ is an assumed symmetric and positive definite matrix. Strict convexity is sufficient for all optimization problems considered in this paper to admit unique closed-form solutions.

On the other hand, $R_P$ is not a coherent risk measure. Recall that a risk measure $\rho({\pmb w}):\mathbb{R}^n\to\mathbb{R}$ is called coherent if it satisfies the following four axioms (adopted from \cite{Artzner1999}):
\begin{itemize}
\item {\bf Monotonicity}. If two portfolios ${\pmb w_1},\,{\pmb w_2}\in\mathcal{P}$ are such that ${\pmb R}{\pmb w}_1^T\leqslant {\pmb R}{\pmb w}_2^T$ almost surely, then $\rho({\pmb w}_1)\geqslant \rho({\pmb w}_2)$. In other words, if one portfolio almost surely dominates another, its risk should not be greater than that of the dominated portfolio.

\item {\bf Translation invariance}. When a risk-free asset is included, we augment the portfolio vector to
${\pmb w}_F=(w_1,\,\ldots,w_n,w_f)\in\mathbb R^{n+1}$, where $w_f$ denotes the proportion invested in the risk-free asset. Then, for every additional investment $a\in\mathbb R$ in the risk-free asset and vector ${\pmb e}_F=(0,\,0,\,\ldots,\,0,\,1)\in\mathbb{R}^{n+1}$, it holds
$\rho({\pmb w}_F+a{\pmb e}_F)=\rho({\pmb w}_F)-ar_f$, where $r_f$ is a risk-free return rate. In other words, increasing the investment in the risk-free asset by an amount $a>0$ decreases the portfolio risk by the corresponding deterministic return.

\item {\bf Positive homogeneity}. If $\lambda\geqslant0$, then 
$\rho(\lambda{\pmb w})=\lambda\rho({\pmb w})$. In other words, scaling the portfolio by a factor $\lambda$ scales its risk by the same factor.

\item {\bf Subadditivity}. If ${\pmb w}_1 + {\pmb w}_2 \in \mathcal{P}$  then
$\rho({\pmb w}_1+{\pmb w}_2)\leqslant \rho({\pmb w}_1) + \rho({\pmb w}_2)
$. Meaning that diversification should not increase risk.
\end{itemize}

Unlike the standard deviation, $R_P$ is not, in general, positively homogeneous. The loss of this property leads to the separation of the tangency portfolio from the maximum Sharpe ratio portfolio and the replacement of the Capital Market Line by a nonlinear Capital Market Curve; see Proposition \ref{P5} and comment beneath Proposition \ref{P6}.

Although the proposed quadratic risk measure is generally not coherent, this is intentional. Positive homogeneity assumes that risk scales proportionally with portfolio size. This assumption may fail in the presence of transaction costs, liquidity effects, taxes, or other nonlinear features of financial markets. The additional linear and constant terms, {\pmb c} and $k$ in \eqref{mean_var_def} respectively, provide flexibility to model such effects while not losing the tractability of quadratic optimization.

\subsection{Possible applications of the generalized quadratic risk measure}

In this subsection, we consider possible choices of $Q$ and ${\pmb c}$ in the quadratic risk measure $R_P$ \eqref{mean_var_def}. Let us start with the matrix $Q$. Sample covariance matrices are often unstable, particularly when the number of assets is large relative to the available sample size \cite{LedoitWolf2004}. One common solution is to shrink the sample covariance matrix toward a more stable, structured matrix
\begin{align*}
Q=(1-\lambda)\Sigma+\lambda I, \quad 0\leqslant\lambda \leqslant 1
\end{align*}
where $I$ is the identity matrix or, more generally,
\begin{align*}
Q=(1-\lambda)\Sigma+\lambda T, \quad 0\leqslant\lambda \leqslant 1,
\end{align*}
where $T$ is a prescribed symmetric and positive definite target matrix; see, for example, \cite{LedoitWolf2003}, \cite{LedoitWolf2004}, \cite{SchaferStrimmer2005}. 

Alternatively, the matrix $Q$ may be obtained from covariance (or scatter) estimators, such as the Minimum Covariance Determinant estimator \cite{RousseeuwVanDriessen1999}, Tyler's
$M$-estimator \cite{Tyler1987}, or Huber's $M$-estimator
\cite{HuberRonchetti2009}. Such estimators reduce the influence of
outliers and provide more robust estimates of the dependence structure
than the sample covariance matrix; see also the monograph
\cite{Maronna2019}.

Another possibility is to construct the matrix $Q$ using a factor
model, which, instead of estimating the full covariance matrix directly,
assumes that a small number of common risk
factors drive asset returns. Classic
examples include the Capital Asset Pricing Model and the
Fama--French three-factor model \cite{FamaFrench1993}. Under appropriate factor model assumptions, $Q$ admits the decomposition
\[
Q=\beta\Sigma_F\beta^T+D,
\]
where $\beta$ is the factor loading matrix,
$\Sigma_F$ is the covariance matrix of the common risk factors, and
$D$ is a diagonal matrix of idiosyncratic variances. Such
representations substantially reduce the number of parameters that
must be estimated while preserving positive definiteness; see, for
example, \cite{FanLiaoMincheva2011}, \cite[Chap.~6]{Kissell2014}.

Let us turn to the linear adjustment vector {\pmb c}. The linear term has a clear geometric interpretation. Since $Q$ is positive definite, the risk measure $R_P$ \eqref{mean_var_def} can be written in the completed-square form \eqref{risk_non_neg}. Consequently, the vector ${\pmb c}$ shifts the center of the quadratic risk surface from the origin to $-Q^{-1}{\pmb c}$. In the classical variance-based model, the investor is penalized solely for portfolio exposure measured relative to the origin. Under the proposed quadratic risk measure, the investor is penalized for deviations from the reference portfolio $-Q^{-1}{\pmb c}$. Depending on the application, this reference portfolio may represent a benchmark allocation,
transaction-cost adjustment, or another preferred exposure. Let us consider three examples.

First, the classical Markowitz model assumes frictionless
trading and therefore ignores the cost of changing portfolio weights.
If ${\pmb w}_{\rm old}$ and ${\pmb w}_{\rm new}$ denote the current and
new portfolios, respectively, a common quadratic transaction cost model (see \cite{BoydConvexPortfolio}, \cite{Peng2011} or \cite{Wang2008})
is
\begin{align}\label{to_expand}
({\pmb w}_{\rm new}-{\pmb w}_{\rm old})
\Lambda
({\pmb w}_{\rm new}-{\pmb w}_{\rm old})^T,
\end{align}
where the positive definite and symmetric matrix $\Lambda$ measures trading cost intensity or market impact. Expanding the quadratic form \eqref{to_expand} yields
\begin{align*}
{\pmb w}_{\rm new}\Lambda{\pmb w}_{\rm new}^T
-
2{\pmb w}_{\rm old}\Lambda{\pmb w}_{\rm new}^T
+
{\pmb w}_{\rm old}\Lambda{\pmb w}_{\rm old}^T,
\end{align*}
which is precisely of the form \eqref{mean_var_def} with
\begin{align*}
Q=2\Lambda,\qquad
{\pmb c}=-2{\pmb w}_{\rm old}\Lambda,\qquad
k={\pmb w}_{\rm old}\Lambda{\pmb w}_{\rm old}^T.
\end{align*}

A second application is benchmark-relative portfolio optimization; see \cite{DeNardLedoitWolf2025}.
Given a benchmark portfolio ${\pmb b}\in\mathcal P$, one often minimizes
the tracking error variance
\[
({\pmb w}-{\pmb b})
\Sigma
({\pmb w}-{\pmb b})^T,
\]
which again expands into a quadratic function of the form
\eqref{mean_var_def}. 

A third application is regularized portfolio optimization; see \cite{DeMiguel2009}, \cite{ZhaoKongQi2021}. Since the classical Markowitz theory is sensitive to estimation errors in expected returns and the covariance matrix $\Sigma$, one often adds a quadratic
penalty to stabilize the optimization problem. For example, the appended quadratic penalty $\lambda {\pmb w}{\pmb w}^T$ leads to the modified risk measure
\begin{align*}
{\pmb w}\Sigma{\pmb w}^T+\lambda {\pmb w}{\pmb w}^T
=
{\pmb w}(\Sigma+\lambda I){\pmb w}^T,
\end{align*}
which corresponds to
\[
Q=2(\Sigma+\lambda I),
\qquad
{\pmb c}={\pmb0},
\qquad
k=0.
\]
More generally, regularization toward a target portfolio
${\pmb w}_T$ yields,
\begin{align*}
{\pmb w}\Sigma{\pmb w}^T+\lambda\left({\pmb w}-{\pmb w}_T\right)\left({\pmb w}-{\pmb w}_T\right)^T,
\end{align*}
which also fits into the proposed framework, producing suitable values of
$Q$, ${\pmb c}$ and $k$.

The proposed quadratic risk measure $R_P$ \eqref{mean_var_def} inherits one important limitation of variance: it is symmetric, penalizing positive and negative deviations
in the same quadratic manner. Consequently, it does not distinguish upside volatility from downside risk. On the other hand, quadratic optimization extends beyond portfolio theory. Positive definite quadratic forms play a central role in optimal control through
Lyapunov functions and linear-quadratic regulators
\cite{Boyd1994},\cite{Khalil2002}, in machine learning through Mahalanobis
distance metrics and metric learning \cite{Weinberger2009}, and in
physics and medicine through quadratic energy functionals and the
linear-quadratic model of radiation response
\cite{Fowler1989}, \cite{McMahon2019}. Therefore, although motivated by
portfolio optimization, the analytical results developed in this paper
may also be relevant in a broader optimization context.

\section{Main results}\label{sec:results}

In this section, we present the main results, showing that the main results of Markowitz portfolio selection theory remain available in explicit closed form under the proposed quadratic risk measure. Proofs are given in the later section.

\begin{proposition}[Minimum quadratic-risk portfolio]\label{P1}
Let $P={\pmb R}{\pmb w}^T$, where ${\pmb w}\in\mathbb{R}^n$ and
${\pmb R}=(r_1,r_2,\,\ldots,\,r_n)$ is the random return vector. The portfolio minimizing the generalized quadratic risk measure $R_P$, subject to the budget constraint ${\pmb 1}{\pmb w}^T=1$, is given by
\begin{align}\label{P1_Min_VaR_Portf}
{\pmb w}^T
=
\frac{1+{\pmb 1}Q^{-1}{\pmb c}^T}
{{\pmb 1}Q^{-1}{\pmb 1}^T}
Q^{-1}{\pmb 1}^T
-
Q^{-1}{\pmb c}^T.
\end{align}
The corresponding minimum value is
\begin{align}\label{P1:minimum}
\min_{{\pmb w}\in\mathcal P}R_P({\pmb w})
=
\frac{
\left(1+{\pmb 1}Q^{-1}{\pmb c}^T\right)^2
-
\left({\pmb 1}Q^{-1}{\pmb 1}^T\right)
\left({\pmb c}Q^{-1}{\pmb c}^T\right)
}
{
2\left({\pmb 1}Q^{-1}{\pmb 1}^T\right)
}
+k.
\end{align}
\end{proposition}

\begin{proposition}[Maximum Sharpe ratio portfolio]\label{P2}
Let $P={\pmb R}{\pmb w}^T$, where ${\pmb w}\in\mathbb{R}^n$ and
${\pmb R}=(r_1,r_2,\,\ldots,\,r_n)$ is the random return vector. Suppose ${\pmb c}Q^{-1}{\pmb c}^T<2k$, which guarantees $R_P({\pmb w})>0$ for all ${\pmb w}\in\mathcal{P}$. The portfolio ${\pmb w}$ maximizing the Sharpe ratio
\begin{align*}
S_P({\pmb w})
=
\frac{{\pmb w}\tilde{\pmb\mu}^T}
{\sqrt{R_P({\pmb w})}}
=
\frac{{\pmb w}{\pmb\mu}^T-r_f}
{\sqrt{\frac12{\pmb w}Q{\pmb w}^T+{\pmb c}{\pmb w}^T+k}},
\end{align*}
subject to the budget constraint ${\pmb 1}{\pmb w}^T=1$, is given by
\begin{align}\label{P2:w}
{\pmb w}^T
=
Q^{-1}
\left(
\frac{A(E-2k)-(C+1)^2}
{AD-B(C+1)}
\,\tilde{\pmb\mu}^T
+
\frac{D(C+1)-B(E-2k)}
{AD-B(C+1)}
\,{\pmb1}^T
-
{\pmb c}^T
\right),
\end{align}
provided that $AD-B(C+1)\neq0$, which implies $\left(1+{\pmb c}Q^{-1}{\pmb 1}^T\right){\pmb \mu}^T \neq \left(r_f + {\pmb c}Q^{-1}{\pmb \mu}^T\right){\pmb 1}^T$,
where
\begin{align*}
A&={\pmb1}Q^{-1}{\pmb1}^T,
&
B&={\pmb1}Q^{-1}\tilde{\pmb\mu}^T,
\\
C&={\pmb c}Q^{-1}{\pmb1}^T,
&
D&={\pmb c}Q^{-1}\tilde{\pmb\mu}^T,
&
E&={\pmb c}Q^{-1}{\pmb c}^T.
\end{align*}

The corresponding maximum value is
\begin{align}\label{P2:maxSR}
\max_{{\pmb w}\in\mathcal P}S_P({\pmb w})
=
\operatorname{sgn}\!\left(B(1+C)-AD\right)
\sqrt{\frac{2\Omega}{\Delta}},
\end{align}
where
\begin{align*}
\Delta&=(1+C)^2-A(E-2k)>0,\\
\Omega&=(2k-E)(AF-B^2)
      +F(1+C)^2
      -2BD(1+C)
      +AD^2>0,\\
F&=\tilde{\pmb\mu}Q^{-1}\tilde{\pmb\mu}^T.
\end{align*}
\end{proposition}

\begin{proposition}[Minimum quadratic-risk portfolio with prescribed expected return]\label{P3}
Let $P={\pmb R}{\pmb w}^T$, where ${\pmb w}\in\mathbb{R}^n$,
${\pmb R}=(r_1,r_2,\,\ldots,\,r_n)$ is the random return vector, and
${\pmb \mu}\neq{\pmb 0}$. Then the generalized quadratic risk measure
\begin{align*}
R_P({\pmb w})
=
\frac12{\pmb w}Q{\pmb w}^T
+
{\pmb c}{\pmb w}^T
+
k,
\end{align*}
subject to the constraints
$
{\pmb \mu}{\pmb w}^T=\mu_0
$,
$
{\pmb 1}{\pmb w}^T=1,
$
attains its minimum at
\begin{align}\label{P3:w}
{\pmb w}^T
=
\left(
Q^{-1}{\pmb \mu}^T,\,
Q^{-1}{\pmb 1}^T
\right)
\begin{pmatrix}
{\pmb \mu}Q^{-1}{\pmb \mu}^T &
{\pmb \mu}Q^{-1}{\pmb 1}^T
\\
{\pmb \mu}Q^{-1}{\pmb 1}^T &
{\pmb 1}Q^{-1}{\pmb 1}^T
\end{pmatrix}^{-1}
\begin{pmatrix}
\mu_0+{\pmb \mu}Q^{-1}{\pmb c}^T
\\
1+{\pmb 1}Q^{-1}{\pmb c}^T
\end{pmatrix}
-
Q^{-1}{\pmb c}^T.
\end{align}
\end{proposition}

Analogously to the classical mean--variance framework, the efficient frontier is defined as the graph of the minimum achievable risk for each predefined expected return. Proposition~\ref{P3} yields its explicit equation. Let $\mu_{R\min}$ denote the expected return of the minimum-risk portfolio \eqref{P1_Min_VaR_Portf}, i.e.
\begin{align*}
\mu_{R\min}
=
{\pmb \mu}\left(
\frac{1+{\pmb 1}Q^{-1}{\pmb c}^T}
{{\pmb 1}Q^{-1}{\pmb 1}^T}
Q^{-1}{\pmb 1}^T
-
Q^{-1}{\pmb c}^T
\right).
\end{align*}

\begin{proposition}[Efficient frontier of risky assets]\label{P4}
Let $P={\pmb R}{\pmb w}^T$, where ${\pmb w}\in\mathbb{R}^n$ and
${\pmb R}=(r_1,r_2,\,\ldots,\,r_n)$ is the random return vector. Then the efficient frontier of risky assets is given by
\begin{align}\label{P4:eff:frontier}
R_P
&=
\frac{a}{2d}\mu^2
+
\frac{ae-b(1+f)}{d}\mu
+
\frac{ae^2-2be(1+f)+c(1+f)^2}{2d}
-\frac{g}{2}
+k,
\,
\mu\geqslant\mu_{R\min},
\end{align}
where $\mu_{R\,min}=-e+b(1+f)/a$,
\begin{align*}
a&={\pmb1}Q^{-1}{\pmb1}^T,
&
b&={\pmb1}Q^{-1}{\pmb\mu}^T,
&
c&={\pmb\mu}Q^{-1}{\pmb\mu}^T,
\\
e&={\pmb\mu}Q^{-1}{\pmb c}^T,
&
f&={\pmb1}Q^{-1}{\pmb c}^T,
&
g&={\pmb c}Q^{-1}{\pmb c}^T,
\end{align*}
and
\begin{align*}
d
=
({\pmb\mu}Q^{-1}{\pmb\mu}^T)
({\pmb1}Q^{-1}{\pmb1}^T)
-
({\pmb\mu}Q^{-1}{\pmb1}^T)^2
>0,
\end{align*}
provided that the vectors ${\pmb\mu}$ and ${\pmb 1}$ are non-colinear.
\end{proposition}

\begin{cor}\label{C}
Identity \eqref{P4:eff:frontier} can be rewritten as
\begin{align*}
\mu
=
\mu_{R\min}
+
\sqrt{\frac{2d}{a}}
\sqrt{R_P-R_{P \min}},
\qquad
R_P\geqslant R_{P \min},
\end{align*}
where
\begin{align*}
\mu_{R\min}=
-e
+
\frac{b(1+f)}{a},\qquad
R_{P \min}
=
\frac{
(1+f)^2-ag
}
{2a}+k
\end{align*}
is the return and risk of the global minimum risk portfolio \eqref{P1_Min_VaR_Portf}.
\end{cor}

The condition $\mu\geqslant \mu_{R \min}$ in \eqref{P4:eff:frontier} gives the upper branch of the hyperbola on the risk-return plane. If $\mu\in\mathbb{R}$ there, then the provided hyperbola is called the minimum risk frontier.

Suppose that a risk-free asset with return $r_f$ is available. A portfolio is then determined by the weights
$
{\pmb w}_F=({\pmb w},\,w_f)\in\mathbb R^{n+1},
$
where
$
w_f=1-{\pmb1}{\pmb w}^T
$
denotes the proportion invested in the risk-free asset. We extend the generalized quadratic risk measure for such portfolios by
\begin{align}\label{grm}
R_P({\pmb w}_F)
=
\frac12{\pmb w}Q{\pmb w}^T
+
{\pmb c}{\pmb w}^T
+
c_fw_f
+
k,
\end{align}
where $c_f\in\mathbb R$ represents the linear contribution of the risk-free asset. Since
$w_f=1-{\pmb1}{\pmb w}^T$,
it follows that
\begin{align*}
R_P({\pmb w}_F)
=\frac12{\pmb w}Q{\pmb w}^T
+\tilde{\pmb c}{\pmb w}^T
+\tilde{k},
\end{align*}
where $\tilde{\pmb c}:={\pmb c}-c_f{\pmb1}$ and $\tilde k:=k+c_f$.

Furthermore, the expected excess return of the portfolio with a risk-free investment is
\begin{align*}
{\pmb w}{\pmb \mu}^T+w_fr_f-r_f
=
{\pmb w}(\pmb\mu-r_f{\pmb1})^T
=
{\pmb w}\tilde{\pmb\mu}^T.
\end{align*}

Therefore, minimizing the generalized quadratic risk measure with risk-free investment with a prescribed excess return is equivalent to solving
\begin{align}\label{equivalent}
\min_{{\pmb w}\in\mathbb R^n}
\quad
\frac12{\pmb w}Q{\pmb w}^T
+
\tilde{\pmb c}{\pmb w}^T
+
\tilde k\quad
\text{subject to}\quad
{\pmb w}\tilde{\pmb\mu}^T=\tilde\mu_0.
\end{align}

\begin{proposition}[Efficient frontier with the risk-free asset]\label{P5}
Let $P=w_1r_1+\cdots+w_nr_n+w_fr_f$, where
$w_f=1-{\pmb1}{\pmb w}^T$.
Then the efficient frontier with the risk-free asset is given by
\begin{align}\label{cmc}
\mu_P
=
r_f
+
\sqrt{
\left(
2R_P({\pmb w}_F)
+
\tilde{\pmb c}Q^{-1}\tilde{\pmb c}^{\,T}
-
2\tilde{k}
\right)
\left(
\tilde{\pmb\mu}Q^{-1}\tilde{\pmb\mu}^{\,T}
\right)
}
-
\tilde{\pmb\mu}Q^{-1}\tilde{\pmb c}^{\,T},
\end{align}
provided that
\begin{align*}
R_P({\pmb w}_F)
\geqslant
\tilde{k}
-
\frac12
\tilde{\pmb c}Q^{-1}\tilde{\pmb c}^{\,T}.
\end{align*}
\end{proposition}

\begin{proposition}[Tangency of the efficient frontiers]\label{P6}
Suppose that ${\pmb1}Q^{-1}\tilde{\pmb\mu}^{T}\neq0$ and 
\begin{align*}
\frac{1+{\pmb 1}Q^{-1}\tilde{\pmb c}^T}
{{\pmb1}Q^{-1}\tilde{\pmb\mu}^{T}}
>0
\end{align*}

Then the efficient frontiers given in \eqref{P4:eff:frontier} and \eqref{cmc} have a unique common point at the corresponding portfolio
\begin{align}\label{tangent_P}
{\pmb w}^T
=
Q^{-1}
\left(
\frac{
1+{\pmb1}Q^{-1}\tilde{\pmb c}^{\,T}
}
{
{\pmb1}Q^{-1}\tilde{\pmb\mu}^{\,T}
}
\tilde{\pmb\mu}^{\,T}
-
\tilde{\pmb c}^{\,T}
\right),
\end{align}
and these two efficient frontiers are tangent at this point.
\end{proposition}

Unlike in the classical mean-variance case, the tangent portfolio \eqref{tangent_P} and the maximum Sharpe ratio portfolio \eqref{P2:w} are not the same. Moreover, identity \eqref{cmc} is no longer a line; so the generalized quadratic risk measure gives rather a Capital Market Curve \eqref{cmc} instead of the well-known Capital Market Line, which is implied by \eqref{cmc} with $\tilde{\pmb c}={\pmb 0}$ and $\tilde{k}=0$. See also \cite{MondalSelvaraju2019} for the illustration of the Capital Market Curve appearance in the mean--lower partial moment framework.

\begin{proposition}[Generalized mutual fund separation theorem]\label{P7}
Let ${\pmb w}_1,\,{\pmb w}_2,\,\ldots,\,{\pmb w}_m\in\mathcal{P}$ be portfolios lying on the minimal risk frontier with expected returns
$\mu_{0,\,1},\,\mu_{0,\,2},\,\ldots,\,\mu_{0,\,m}$, and linear adjustment vectors ${\pmb c}_1,\,{\pmb c}_2,\,\ldots,\,{\pmb c}_m$ respectively. Let
$
(\lambda_1,\lambda_2,\,\ldots,\lambda_m)\in\mathbb{R}^m
$
and $\mu_0\in\mathbb{R}$ satisfy
\begin{align}\label{syst}
\begin{cases}
\lambda_1+\lambda_2+\cdots+\lambda_m=1\\
\lambda_1\mu_{0,\,1}
+\lambda_2\mu_{0,\,2}
+\cdots
+\lambda_m\mu_{0,\,m}
=
\mu_0
\end{cases}.
\end{align}
Then the portfolio $
\lambda_1{\pmb w}_1
+\lambda_2{\pmb w}_2
+\cdots
+\lambda_m{\pmb w}_m
$ with the linear adjustment vector ${\pmb c}=\lambda_1{\pmb c}_1+\lambda_2{\pmb c}_2+\cdots+\lambda_m{\pmb c}_m$ belongs to the minimum risk frontier also. If, in addition, $\mu_0\geqslant\mu_{\min}$, where $\mu_{\min}$ denotes the expected return of the minimum risk portfolio, computed by \eqref{P1_Min_VaR_Portf}, whose linear adjustment vector is ${\pmb c}$, then such a portfolio is efficient.
\end{proposition}

{\sc Remark 1:}
{\it
Because of \eqref{grm} and \eqref{equivalent}, Proposition~\ref{P7} remains valid for portfolios with a risk-free asset after replacing the asset-return vector ${\pmb\mu}$ by $\tilde{\pmb\mu}={\pmb\mu}-r_f{\pmb1}$, the linear adjustment vector ${\pmb c}$ by $\tilde{\pmb c}$, and the prescribed expected returns by the corresponding prescribed excess returns.
}

Proposition \eqref{P7}, unlike the classical mutual fund separation theorem \cite[Prop. 7]{GG2024}, states that the constituent efficient portfolios may correspond to different linear adjustment vectors. This algebraic property does not rely on assumptions of distributions of random returns. On the other hand, such assumptions are important to prove analogous separation statements in non-quadratic risk measures like value at risk or conditional values at risk; see, for example, \cite{Cass1970}, \cite{Owen1983}, \cite{Ross1978}. See also \cite{DeGiorgi2011} for sufficient conditions for two-fund separation in general reward-risk portfolio models or \cite{AlexanderBaptista2007} on how fund separation changes when alternative risk constraints are added to the Markowitz mean–variance model.

\begin{proposition}[Maximum utility portfolio]\label{P8}
Let $P={\pmb R}{\pmb w}^T$, where ${\pmb w}\in\mathbb{R}^n$ and
${\pmb R}=(r_1,r_2,\,\ldots,\,r_n)$ is the random return vector. Then the utility function
\begin{align*}
U({\pmb w})
=
{\pmb w}{\pmb\mu}^T
-
\left(
\frac12{\pmb w}Q{\pmb w}^T
+
{\pmb c}{\pmb w}^T
+
k
\right),
\end{align*}
subject to the budget constraint
$
{\pmb1}{\pmb w}^T=1,
$
attains its maximum at
\begin{align}\label{max_u}
{\pmb w}^T
=Q^{-1}\left(
\frac{
1+{\pmb1}Q^{-1}{\pmb c}^T
-
{\pmb1}Q^{-1}{\pmb\mu}^T
}
{{\pmb1}Q^{-1}{\pmb1}^T}
{\pmb1}^T
+
{\pmb\mu}^T
-
{\pmb c}^T
\right).
\end{align}
The corresponding maximum value is
\begin{align*}
\max_{{\pmb w}\in\mathcal P}U({\pmb w})
=\frac12\left(
c-2e+g-\frac{(b-f-1)^2}{a}
\right)-k.
\end{align*}
where
\begin{align*}
a&={\pmb1}Q^{-1}{\pmb1}^T,
&
b&={\pmb1}Q^{-1}{\pmb\mu}^T,
&
c&={\pmb\mu}Q^{-1}{\pmb\mu}^T,
\\
e& = {\pmb \mu}Q^{-1}{\pmb c}^T,
&
f&={\pmb1}Q^{-1}{\pmb c}^T,
&
g&={\pmb c}Q^{-1}{\pmb c}^T.
\end{align*}
Moreover, the portfolio \eqref{max_u} is efficient.
\end{proposition}

\begin{proposition}[Maximum utility portfolio with risk-free asset]\label{P9}
    Let $P={\pmb R}{\pmb w}_F^T$, where ${\pmb w}_F\in\mathbb{R}^{n+1}$ is a portfolio and
${\pmb R}=(r_1,r_2,\,\ldots,\,r_n,\,r_f)$ is the random return vector with a risk-free rate. Then the utility function
\begin{align*}
U({\pmb w}_F)
= {\pmb \mu}{\pmb w}^T+w_fr_f - R_P(\pmb{w}_F) = {\pmb w}\tilde{\pmb \mu}^T +r_f- 
\left(
\frac12{\pmb w}Q{\pmb w}^T
+
\tilde{\pmb c}{\pmb w}^T
+
\tilde{k}
\right),
\end{align*}
where $\tilde{{\pmb \mu}}={\pmb \mu}-r_f{\pmb 1}$, $\tilde{\pmb c}={\pmb c}-c_f{\pmb 1}$, and $\tilde{k}=k+c_f$, attains its maximum at
\begin{align}\label{max_u_rf}
{\pmb w}^T
=Q^{-1}\left(
\tilde{\pmb \mu} - \tilde{\pmb c}
\right)^T,\, w_f = 1-{\pmb 1}{\pmb w}^T
\end{align}
and
\begin{align*}
\max_{{\pmb w}\in\mathcal P}U({\pmb w})
= r_f -\tilde{k} + \frac{p+q-2s}{2},
\end{align*}
where
\begin{align*}
p&=\tilde{\pmb\mu}Q^{-1}\tilde{\pmb\mu}^T,
&
q&=\tilde{\pmb c}Q^{-1}\tilde{\pmb c}^T,
&
s&=\tilde{\pmb\mu}Q^{-1}\tilde{\pmb c}^T.
\end{align*}
Moreover, the portfolio \eqref{max_u_rf} is efficient.
\end{proposition}

\section{Proofs}\label{sec:proofs}

In this section, we prove the statements formulated in Section \ref{sec:results}.

\begin{proof}[Proof of Proposition \ref{P1}]

Let us set up the Lagrangian function
\begin{align*}
L\left({\pmb w},\,\lambda\right)=\frac{1}{2}{\pmb w}Q{\pmb w}^T+{\pmb c}{\pmb w}^T + k-\lambda\left({\pmb 1}{\pmb w}^T-1\right).
\end{align*}
By computing its partial derivatives with respect to $w_1,\,w_2,\,\ldots,\,w_n$ and $\lambda$, we obtain the following system of linear equations
\begin{align}\label{proof:P1:syst}
\begin{cases}
Q{\pmb w}^T+{\pmb c}^T-\lambda{\pmb 1}^T={\pmb 0}^T \\
{\pmb 1}{\pmb w}^T=1
\end{cases}.
\end{align}
Multiplying the first equation in \eqref{proof:P1:syst} by $Q^{-1}$ we obtain
\begin{align}\label{proof:P1:w1}
    {\pmb w}^T = Q^{-1}\left(\lambda{\pmb 1}^T-{\pmb c}^T\right).
\end{align}
By inserting this expression into the second equation of the system \eqref{proof:P1:syst}, we get
\begin{align*}
\lambda{\pmb 1}Q^{-1}{\pmb 1}^T-{\pmb 1}Q^{-1}{\pmb c}^T = 1\qquad \implies \qquad \lambda = \frac{1+{\pmb 1}Q^{-1}{\pmb c}^T}{{\pmb 1}Q^{-1}{\pmb 1}^T}.
\end{align*}
Substituting this value of $\lambda$ into \eqref{proof:P1:w1} gives
\begin{align}\label{proof:P1:w2}
    {\pmb w}^T = \frac{1+{\pmb 1}Q^{-1}{\pmb c}^T}{{\pmb 1}Q^{-1}{\pmb 1}^T}Q^{-1}{\pmb 1}^T - Q^{-1}{\pmb c}^T.
\end{align}

Since $R_P$ is strictly convex, every stationary point is the unique global minimizer. Therefore, \eqref{proof:P1:w2} is the unique minimizer.

Finally, evaluating $R_P$ at the minimizer \eqref{proof:P1:w2} and using \eqref{risk_non_neg} yields
\begin{align*}
\min\limits_{{\pmb w}\in\mathcal{P}}R_P({\pmb w}) = \frac{1}{2}\frac{\left(1+{\pmb 1}Q^{-1}{\pmb c}^T\right)^2}{{\pmb 1}Q^{-1}{\pmb 1}^T}-\frac{1}{2}{\pmb c}Q^{-1}{\pmb c}^T+k.
\end{align*}
\end{proof}

\begin{proof}[Proof of Proposition \ref{P2}]
Let $\mathcal{R}$ denote the quadratic risk measure $R_P({\pmb w})$  \eqref{mean_var_def} evaluated at ${\pmb w}$. Throughout this proof, we assume that
\begin{align*}
k>\frac12{\pmb c}Q^{-1}{\pmb c}^T,
\end{align*}
which implies $\mathcal{R}=R_P({\pmb w})>0$ for every ${\pmb w}\in\mathcal P$. Hence, the Sharpe ratio is well defined.

We set up the Lagrangian function
\begin{align*}
L({\pmb w},\lambda)
=\frac{{\pmb w}\boldsymbol{\mu}^T-r_f}{\sqrt{\frac{1}{2}{\pmb w}Q{\pmb w}^T+{\pmb c}{\pmb w}^T+k}}-\lambda\left({\pmb 1}{\pmb w}^T-1\right)=\frac{{\pmb w}{\mu}^T-r_f}{\sqrt{\mathcal{R}}}-\lambda\left({\pmb 1}{\pmb w}^T-1\right).
\end{align*}

By computing its partial derivatives with respect to $w_1,\,w_2,\,\ldots,\,w_n$ and $\lambda$, we obtain
\begin{align}\label{proof:P2:system}
\begin{cases}
{\pmb\mu}^T\mathcal{R}
-\dfrac12({\pmb w}{\pmb\mu}^T-r_f)
(Q{\pmb w}^T+{\pmb c}^T)
-\lambda{\pmb1}^T\mathcal{R}^{3/2}
={\pmb0}^T\\
{\pmb1}{\pmb w}^T=1
\end{cases}.
\end{align}

Left-multiplying the first equation in \eqref{proof:P2:system} by ${\pmb w}\in\mathbb{R}^n$ and using the budget constraint ${\pmb 1}{\pmb w}^T=1$ gives
\begin{align*}
&{\pmb w}{\pmb\mu}^T\mathcal{R}
-\dfrac12({\pmb w}{\pmb\mu}^T-r_f)
({\pmb w}Q{\pmb w}^T+{\pmb c}{\pmb w}^T)
-\lambda\mathcal{R}^{3/2}
=0
\end{align*}
and, consequently,
\begin{align}\label{proof:P2:lambda:form}
\lambda
=
\frac{r_f}{\sqrt{\mathcal R}}
+
\frac{
({\pmb w}{\pmb\mu}^T-r_f)
\left({\pmb c}{\pmb w}^T+2k\right)
}
{2\mathcal R^{3/2}}.
\end{align}

Substituting \eqref{proof:P2:lambda:form} into the first equation of \eqref{proof:P2:system} and simplifying, yields

\begin{align}\label{proof:P2:eq1}
\frac{{\pmb w}{\pmb\mu}^T-r_f}{2}
\left(
Q{\pmb w}^T
+
{\pmb c}^T
+
{\pmb1}^T({\pmb c}{\pmb w}^T+2k)
\right)
=
\mathcal R\,\tilde{\pmb\mu}^{\,T},
\end{align}
where $\tilde{{\pmb \mu}}^T={\pmb \mu}^T-r_f{\pmb 1}^T$.

Multiplying \eqref{proof:P2:eq1} by $Q^{-1}$ we obtain

\begin{align}\label{proof:P2:w1}
{\pmb w}^T
=
Q^{-1}
\left(
\frac{2\mathcal R}
{{\pmb w}{\pmb\mu}^T-r_f}
\tilde{\pmb\mu}^{\,T}
-
\left({\pmb c}{\pmb w}^T+2k\right){\pmb1}^T
-
{\pmb c}^T
\right).
\end{align}

Let us denote 
\begin{align}\label{alpha_def}
\alpha:=\frac{2\mathcal{R}}{{\pmb w}{\pmb \mu}^T - r_f},\,
\beta:={\pmb c}{\pmb w}^T+2k.
\end{align}
Then, by left-multiplying \eqref{proof:P2:w1} by ${\pmb 1}$ and ${\pmb c}$ respectively, we obtain the following linear system for the unknowns $\alpha$ and $\beta$
\begin{align}\label{system_2x2}
\begin{pmatrix}
B&-A\\
D&-(C+1)
\end{pmatrix}
\begin{pmatrix}
\alpha \\
\beta
\end{pmatrix}
=
\begin{pmatrix}
C+1\\
E-2k
\end{pmatrix},
\end{align}
where 
\begin{align*}
A&={\pmb1}Q^{-1}{\pmb1}^T,
&
B&={\pmb1}Q^{-1}\tilde{\pmb\mu}^T,
\\
C&={\pmb c}Q^{-1}{\pmb1}^T,
&
D&={\pmb c}Q^{-1}\tilde{\pmb\mu}^T,
&
E&={\pmb c}Q^{-1}{\pmb c}^T.
\end{align*}

The matrix in \eqref{system_2x2} is invertible whenever $\text{det}:=AD-B(C+1)\neq 0$, which implies $\left(1+{\pmb c}Q^{-1}{\pmb 1}^T\right){\pmb \mu}^T \neq \left(r_f + {\pmb c}Q^{-1}{\pmb \mu}^T\right){\pmb 1}^T$ since
\begin{align*}
AD-B(C+1)
&=
\left({\pmb1}Q^{-1}{\pmb1}^T\right)
\left({\pmb c}Q^{-1}\tilde{\pmb\mu}^T\right)
-
\left({\pmb1}Q^{-1}\tilde{\pmb\mu}^T\right)
\left({\pmb c}Q^{-1}{\pmb1}^T+1\right)\\
&=
{\pmb1}Q^{-1}
\left(
{\pmb1}^T\left(r_f+{\pmb c}Q^{-1}{\pmb\mu}^T\right)
-
{\pmb\mu}^T\left(1+{\pmb c}Q^{-1}{\pmb1}^T\right)
\right).
\end{align*}

By inserting the solution of \eqref{system_2x2} into \eqref{proof:P2:w1}, we obtain
\begin{align}\nonumber
{\pmb w}^T&=\left(Q^{-1}\tilde{{\pmb \mu}}^T,\,-Q^{-1}{\pmb 1}^T\right)
\begin{pmatrix}
\alpha\\
\beta
\end{pmatrix}
-Q^{-1}{\pmb c}^T\\ \nonumber
&=\frac{1}{\text{det}}\left(Q^{-1}\tilde{{\pmb \mu}}^T,\,-Q^{-1}{\pmb 1}^T\right)
\begin{pmatrix}
-{\pmb c}Q^{-1}{\pmb 1}^T - 1 & {\pmb 1}Q^{-1}{\pmb 1}^T\\
-{\pmb c}Q^{-1}\tilde{{\pmb \mu}}^T & {\pmb 1}Q^{-1}\tilde{{\pmb \mu}}^T
\end{pmatrix}
\begin{pmatrix}
{\pmb 1}Q^{-1}{\pmb c}^T + 1 \\
 {\pmb c}Q^{-1}{\pmb c}^T -2k
\end{pmatrix}
-Q^{-1}{\pmb c}^T\\
&=
Q^{-1}
\left(
\frac{A(E-2k)-(C+1)^2}
{\text{det}}
\,\tilde{\pmb\mu}^T
+
\frac{D(C+1)-B(E-2k)}
{\text{det}}
\,{\pmb1}^T
-
{\pmb c}^T
\right)
.\label{proof:P2:w:final}
\end{align}

We now insert the solution \eqref{proof:P2:w:final} into the Sharpe ratio. This gives

\begin{align}\label{Sharpe_inserted}
S_P({\pmb w})
=
\frac{{\pmb w}{\pmb\mu}^T-r_f}
{\sqrt{R_P({\pmb w})}}
=
\frac{\alpha F+\beta B-D}
{\sqrt{\frac12\alpha^2F+\alpha\beta B+\frac12A\beta^2-\frac12E+k}},
\end{align}
where
\begin{align*}
\alpha=
\frac{A(E-2k)-(C+1)^2}{\mathrm{det}},
\qquad
\beta=
\frac{D(C+1)-B(E-2k)}{\mathrm{det}},\qquad 
F = \tilde{\pmb \mu}Q^{-1}\tilde{\pmb \mu}^T.
\end{align*}

The denominator of \eqref{Sharpe_inserted}, multiplied by $\text{det}^2$, is
\begin{align*}
&\left(\frac12\alpha^2F+\alpha\beta B+\frac12A\beta^2-\frac12E+k
\right)\text{det}^2\\
&=
\frac12F\left(A(E-2k)-(C+1)^2\right)^2
+B\left(A(E-2k)-(C+1)^2\right)
\left(D(C+1)-B(E-2k)\right)
\\
&+
\frac12A\left(D(C+1)-B(E-2k)\right)^2
+\left(k-\frac12E\right)\left(AD-B(C+1)\right)^2\\
&=\frac{1}{2}\Delta\Omega,
\end{align*}
where
\begin{align*}
\Delta:=(C+1)^2-A(E-2k), \quad \Omega:=(2k-E)(AF-B^2)
+F(C+1)^2
-2BD(C+1)
+AD^2.
\end{align*}
For the numerator of \eqref{Sharpe_inserted}, we obtain
\begin{align*}
&(\alpha F+\beta B-D)\text{det}
=
F(A(E-2k)-(C+1)^2)
+
B(D(C+1)-B(E-2k))
-
D\,\mathrm{det}
\\
&=
F(A(E-2k)-(C+1)^2)
+
B(D(C+1)-B(E-2k))
-
D(AD-B(C+1))
\\
&=
(E-2k)(AF-B^2)
-
F(C+1)^2
+
2BD(C+1)
-
AD^2
=
-\Omega,
\end{align*}

One can observe that $\Delta > 0$ since $k>E/2$. Moreover, in the expression of $\Omega$ the first part $(2k-E)(AF-B^2)$ is positive since $k>E/2$ and $AF-B^2> 0$.
The second part in $\Omega$ is $F(C+1)^2-2BD(C+1)+AD^2$ and it is also positive, since it can be rewritten as
\begin{align*}
    \Big(
        C+1 ,\, D
    \Big)
    \begin{pmatrix}
        F & -B \\ -B & A
    \end{pmatrix}
    \begin{pmatrix}
        C+1 \\ D
    \end{pmatrix},
\end{align*}
where the middle matrix is positive definite since $AF-B^2>0$ and $F>0$. 

It remains to prove that the solution \eqref{proof:P2:w:final} indeed maximizes the Sharpe ratio. Let us denote the maximizer from \eqref{proof:P2:w:final} ${\pmb w}_{\max}$ and take an arbitrary feasible portfolio $\tilde{\pmb w}={\pmb w}_{\max}+{\pmb \xi}$, where ${\pmb \xi}\in\mathbb{R}^n$ and ${\pmb 1}{\pmb \xi}^T=0$ because of the budget constraint. Then, the numerator of \eqref{Sharpe_inserted} is 
\begin{align*}
\tilde{{\pmb w}}{\pmb \mu}^T-r_f={\pmb w}_{\max}{\pmb \mu}^T-r_f+{\pmb \xi}{\pmb \mu}^T
={\pmb w}_{\max}{\pmb \mu}^T-r_f+{\pmb \xi}\tilde{{\pmb \mu}}^T
,
\end{align*}
while the risk measure in the denominator of \eqref{Sharpe_inserted} is 
\begin{align*}
R_P(\tilde{\pmb w})=R_P({\pmb w}_{\max})+{\pmb \xi}\left(Q{\pmb w}^T_{\max}+{\pmb c}^T\right)+\frac{1}{2}{\pmb \xi}Q{\pmb \xi}^T.
\end{align*}
Using \eqref{proof:P2:w1},
\begin{align*}
Q{\pmb w}^T_{\max}+{\pmb c}^T=\alpha\tilde{\pmb \mu}^T-\beta {\pmb 1}^T-{\pmb c}^T+{\pmb c}^T
=\alpha\tilde{\pmb \mu}^T-\beta {\pmb 1}^T
\end{align*}
and, consequently,
\begin{align*}
{\pmb \xi}\left(Q{\pmb w}^T_{\max}+{\pmb c}^T\right)=\alpha{\pmb \xi}\tilde{\pmb \mu}^T.
\end{align*}
Therefore,
\begin{align*}
R_P(\tilde{\pmb w})=R_P({\pmb w}_{\max})+\alpha{\pmb \xi}\tilde{\pmb \mu}^T+\frac{1}{2}{\pmb \xi}Q{\pmb \xi}^T.
\end{align*}
Thus, we have to prove
\begin{align}\label{to_prove}
S_p(\tilde{\pmb w})=\frac{{\pmb w}_{\max}{\pmb \mu}^T-r_f+{\pmb \xi}\tilde{{\pmb \mu}}^T}{\sqrt{R_P({\pmb w}_{\max})+\alpha{\pmb \xi}\tilde{\pmb \mu}^T+\frac{1}{2}{\pmb \xi}Q{\pmb \xi}^T}}<
\frac{{\pmb w}_{\max}{\pmb \mu}^T-r_f}{\sqrt{R_P({\pmb w}_{\max})}}=S_P(\pmb w_{\max})
\end{align}
for all ${\pmb \xi}\in\mathbb{R}^n\setminus \{{\pmb 0}\}$ such that ${\pmb 1}{\pmb \xi}^T=0$.

Evaluating \eqref{alpha_def} at ${\pmb w}={\pmb w}_{\max}$, we obtain,
\begin{align*}
\alpha=\frac{2R_P({\pmb w}_{\max})}{{\pmb w}_{\max}{\pmb \mu}^T-r_f}=\frac{2y}{x},
\end{align*}
where $x={\pmb w}_{\max}{\pmb \mu}^T-r_f$, $y=R_P({\pmb w}_{\max})$. By denoting additionally $\varepsilon={\pmb \xi}\tilde{{\pmb \mu}}^T$, $\delta=\frac{1}{2}{\pmb \xi}Q{\pmb \xi}^T$ and looking to \eqref{to_prove}, we determine the inequality
\begin{align*}
\frac{x+\varepsilon}{\sqrt{y+\alpha\varepsilon+\delta}}<\frac{x}{\sqrt{y}}
\qquad \Longleftrightarrow \qquad x+\varepsilon<\sqrt{(x+\varepsilon)^2-\varepsilon^2+\frac{x^2}{y}\delta}.
\end{align*}
The latter inequality is trivial if $x+\varepsilon<0$. Otherwise, it is equivalent to
\begin{align*}
\frac{\varepsilon^2}{\delta}<\frac{x^2}{y}.
\end{align*}

Therefore, observing that $x^2/y=S^2_P({\pmb w}_{\max})=2\Omega/\Delta$, inequality \eqref{to_prove} follows if we can show
\begin{align}\label{to_prove1}
\frac{({\pmb \xi} \tilde{\pmb \mu}^T)^2}{{\pmb \xi}Q {\pmb \xi}^T}<\frac{\Omega}{\Delta}.
\end{align}
for all ${\pmb \xi}\in\mathbb{R}^n\setminus \{{\pmb 0}\}$ such that ${\pmb 1}{\pmb \xi}^T=0$.

Let $\lambda\in\mathbb{R}$. By using ${\pmb 1}{\pmb \xi}^T=0$ and the Cauchy–Schwarz inequality, we obtain
\begin{align*}
&\left({\pmb \xi} \tilde{\pmb \mu}^T\right)^2=({\pmb \xi} (\tilde{\pmb \mu}-\lambda {\pmb 1})^T)^2
\leqslant\left({\pmb \xi}Q{\pmb \xi}^T\right)\left((\tilde{\pmb \mu}-\lambda {\pmb 1})Q^{-1}(\tilde{\pmb \mu}-\lambda {\pmb 1})^T\right)\\
&=\left({\pmb \xi}Q{\pmb \xi}^T\right) \left(F-2\lambda B+\lambda^2A\right).
\end{align*}
Since the quadratic function $F-2\lambda B+\lambda^2A$ is minimal at $\lambda=B/A$, we obtain
\begin{align*}
\frac{({\pmb \xi} \tilde{\pmb \mu}^T)^2}{{\pmb \xi}Q {\pmb \xi}^T}\leqslant F-\frac{B^2}{A}.
\end{align*}
Therefore, inequality \eqref{to_prove1} follows if
\begin{align*}
F-\frac{B^2}{A}<\frac{\Omega}{\Delta},
\end{align*}
but this is true because
\begin{align*}
A\Omega-(AF-B^2)\Delta
&=AF(C+1)^2-2ABD(C+1)+A^2D^2-(AF-B^2)(C+1)^2\\
&=A^2D^2-2ABD(C+1)+B^2(C+1)^2\\
&=(AD-B(C+1))^2=\text{det}^2>0.
\end{align*}
\end{proof}

\begin{proof}[Proof of Proposition \ref{P3}]
    Let us set up the Lagrangian function
\begin{align*}
L({\pmb w},\,\lambda_1,\,\lambda_2)=\frac{1}{2}{\pmb w}Q{\pmb w}^T+{\pmb c}{\pmb w}^T+k-\lambda_1\left({\pmb \mu}{\pmb w}^T-\mu_0\right)-\lambda_2\left({\pmb 1}{\pmb w}^T-1\right)
\end{align*}
and compute its partial derivatives with respect to $w_1,\,w_2,\,\ldots,\,w_n$, $\lambda_1$ and $\lambda_2$:
\begin{align}\label{proof:P3:system}
\begin{cases}
Q{\pmb w}^T + {\pmb c}^T-\lambda_1{\pmb \mu}^T-\lambda_2{\pmb 1}^T={\pmb 0}^T \\
{\pmb \mu}{\pmb w}^T=\mu_0\\
{\pmb 1}{\pmb w}^T=1 
\end{cases}.
\end{align}
Multiplying the first equation of this system by $Q^{-1}$ yields
\begin{align}\label{proof:P3:w1}
    {\pmb w}^T = \lambda_1Q^{-1}{\pmb \mu}^T+\lambda_2Q^{-1}{\pmb 1}^T - Q^{-1}{\pmb c}^T = Q^{-1}
    \left(
        {\pmb \mu}^T,\, {\pmb 1}^T
    \right)
    \begin{pmatrix}
        \lambda_1 \\
        \lambda_2
    \end{pmatrix} - Q^{-1}{\pmb c}^T.
\end{align}
Multiplying \eqref{proof:P3:w1} by vectors ${\pmb \mu}$ and ${\pmb 1}$ respectively, and using the last two equations of \eqref{proof:P3:system}, we obtain
\begin{align*}
\begin{cases}
\lambda_1{\pmb \mu}Q^{-1}{\pmb \mu}^T+\lambda_2{\pmb \mu}Q^{-1}{\pmb 1}^T - {\pmb \mu}Q^{-1}{\pmb c}^T = \mu_0 \\
\lambda_1{\pmb 1}Q^{-1}{\pmb \mu}^T+\lambda_2{\pmb 1}Q^{-1}{\pmb 1}^T - {\pmb 1}Q^{-1}{\pmb c}^T= 1
\end{cases}
\end{align*}
or, in matrix form,
\begin{align}\label{M_def}
\underbrace{\begin{pmatrix}
    {\pmb \mu}Q^{-1}{\pmb \mu}^T & {\pmb \mu}Q^{-1}{\pmb 1}^T \\
    {\pmb \mu}Q^{-1}{\pmb 1}^T & {\pmb 1}Q^{-1}{\pmb 1}^T
\end{pmatrix}}_{=:M}
\begin{pmatrix}
    \lambda_1 \\
    \lambda_2
\end{pmatrix}
=
\begin{pmatrix}
    \mu_0 + {\pmb \mu}Q^{-1}{\pmb c}^T \\
    1 + {\pmb 1}Q^{-1}{\pmb c}^T
\end{pmatrix}.
\end{align}
The matrix $M$ is invertible since
\begin{align*}
\det(M)
=
({\pmb 1}Q^{-1}{\pmb 1}^T)
({\pmb \mu}Q^{-1}{\pmb \mu}^T)
-
({\pmb 1}Q^{-1}{\pmb \mu}^T)^2
>0,
\end{align*}
where the strict inequality follows from the Cauchy--Schwarz inequality because the vectors ${\pmb 1}$ and ${\pmb \mu}$ are non-collinear. Then
\begin{align}\label{proof:P3:lambdas}
    \begin{pmatrix}
    \lambda_1 \\
    \lambda_2
\end{pmatrix} = 
\begin{pmatrix}
    {\pmb \mu}Q^{-1}{\pmb \mu}^T & {\pmb \mu}Q^{-1}{\pmb 1}^T \\
    {\pmb \mu}Q^{-1}{\pmb 1}^T & {\pmb 1}Q^{-1}{\pmb 1}^T
\end{pmatrix}^{-1}
\begin{pmatrix}
    \mu_0 + {\pmb \mu}Q^{-1}{\pmb c}^T \\
    1 + {\pmb 1}Q^{-1}{\pmb c}^T
\end{pmatrix}
\end{align}
and substituting this expression into \eqref{proof:P3:w1} yields
\begin{align}\label{proof:P3:w:final}
    {\pmb w}^T = \left(
        Q^{-1}{\pmb \mu}^T,\, Q^{-1}{\pmb 1}^T
    \right)
    \begin{pmatrix}
    {\pmb \mu}Q^{-1}{\pmb \mu}^T & {\pmb \mu}Q^{-1}{\pmb 1}^T \\
    {\pmb \mu}Q^{-1}{\pmb 1}^T & {\pmb 1}Q^{-1}{\pmb 1}^T
\end{pmatrix}^{-1}
\begin{pmatrix}
    \mu_0+{\pmb \mu}Q^{-1}{\pmb c}^T \\
    1+{\pmb 1}Q^{-1}{\pmb c}^T
\end{pmatrix} -Q^{-1}{\pmb c}^T.
\end{align}

The proof that the portfolio \eqref{proof:P3:w:final} is the unique minimizer is identical to the corresponding argument in the proof of Proposition~\ref{P1}, as the strict convexity of the quadratic risk measure is unaffected by the additional linear constraint.
\end{proof}

\begin{proof}[Proof of Proposition \ref{P4}]
By Proposition~\ref{P3}, the matrix $M$ from \eqref{M_def} is invertible and
\begin{align*}
M^{-1}=\begin{pmatrix}
{\pmb\mu}Q^{-1}{\pmb\mu}^T&{\pmb\mu}Q^{-1}{\pmb 1}^T\\
{\pmb 1}Q^{-1}{\pmb\mu}^T&{\pmb 1}Q^{-1}{\pmb 1}^T
\end{pmatrix}^{-1}
=\frac{1}{d}
\begin{pmatrix}
{\pmb 1}Q^{-1}{\pmb 1}^T&-{\pmb\mu}Q^{-1}{\pmb 1}^T\\
-{\pmb\mu}Q^{-1}{\pmb 1}^T&
{\pmb\mu}Q^{-1}{\pmb\mu}^T
\end{pmatrix},
\end{align*}
where $d=\left({\pmb\mu}Q^{-1}{\pmb\mu}^T\right)\left({\pmb 1}Q^{-1}{\pmb 1}^T\right)-\left({\pmb\mu}Q^{-1}{\pmb 1}^T\right)^2>0$ if ${\pmb 1}$ and ${\pmb \mu}$ are non-collinear. 
For convenience, let
\begin{align*}
a&={\pmb1}Q^{-1}{\pmb1}^T,
&
b&={\pmb\mu}Q^{-1}{\pmb1}^T,
&
c&={\pmb\mu}Q^{-1}{\pmb\mu}^T,
&
d&=ac-b^2,\\
e&={\pmb\mu}Q^{-1}{\pmb c}^T,
&
f&={\pmb1}Q^{-1}{\pmb c}^T,
&
g&={\pmb c}Q^{-1}{\pmb c}^T.
\end{align*}

Additionally, let $\gamma=\mu_0+e$ and $\kappa=1+f$. Substituting \eqref{proof:P3:w:final} into the quadratic form ${\pmb w}Q{\pmb w}^T$, we obtain
\begin{align*}
{\pmb w}Q{\pmb w}^T &= \left(
\frac{1}{d}\left(
\mu_0+e ,\, 1+f
\right)
\begin{pmatrix}
a & -b\\
-b & c
\end{pmatrix}
\left(
     Q^{-1}{\pmb \mu}^T ,\,  Q^{-1}{\pmb 1}^T
\right)^T-{\pmb c}Q^{-1}\right) \\
&\times
\left(
\frac{1}{d}
\left(
{\pmb \mu}^T ,\, {\pmb 1}^T
\right)
\begin{pmatrix}
a & -b\\
-b & c
\end{pmatrix}
\begin{pmatrix}
\mu_0+e \\ 1+f
\end{pmatrix} - {\pmb c}^T
\right)\\
&= \frac{1}{d}\left(a\gamma^2-2b\gamma\kappa+c\kappa^2\right) -\frac{2}{d}\left(\gamma(ae-bf)+\kappa(cf-be)\right) + g.
\end{align*}

Furthermore,
\begin{align*}
    {\pmb c}{\pmb w}^T = \frac{\gamma(ae-bf)+\kappa(cf-be)}{d} - g.
\end{align*}
Hence, 
\begin{align*}
&R_{P}({\pmb w}) = \frac{1}{2}{\pmb w}Q{\pmb w}^T + {\pmb c}{\pmb w}^T + k =\frac{\gamma^2a-2b\gamma\kappa + c\kappa^2}{2d} - \frac{g}{2} + k.
\end{align*}
Finally, substituting the definitions of $\gamma$ and $\kappa$ and simplifying, yields
\begin{align*}
R_{P}({\pmb w})=\frac{a}{2d}\cdot \mu_0^2
+ \frac{ae - b(1+f)}{d}\cdot \mu_0
+ \frac{ae^2 - 2be(1+f) + c(1+f)^2}{2d}
- \frac{g}{2} + k,
\end{align*}
which completes the proof.
\end{proof}

\begin{proof}[Proof of Corolary \ref{C}]
Completing the square in \eqref{P4:eff:frontier} and noting that $b^2 = ac-d$, gives
\begin{align}\label{square}
R_P
&=
\frac{a}{2d}
\left(
\mu+e-\frac{b(1+f)}{a}
\right)^2
+
\frac{
(1+f)^2-ag
}
{2a}+k.
\end{align}
The second term $+k$ on the right-hand side of \eqref{square} equals $R_{P \min}$,
which is the minimum value of the quadratic polynomial. Therefore,
\begin{align*}
R_P-R_{P \min}
=
\frac{a}{2d}
\left(
\mu+e-\frac{b(1+f)}{a}
\right)^2.
\end{align*}
Since the efficient frontier satisfies $\mu\geqslant \mu_{R \min}=-e+b(1+f)/a$,
the positive square-root branch must be chosen. Hence
\begin{align*}
\mu
=
-e+\frac{b(1+f)}{a}
+
\sqrt{\frac{2d}{a}}
\sqrt{R_P-R_{P \min}},
\qquad
R_P\geqslant R_{P\min},
\end{align*}
which proves the statement.
\end{proof}

\begin{proof}[Proof of Proposition \ref{P5}]\
Let us set up the Lagrangian function
\begin{align*}
L({\pmb w},\,\lambda)
=
\frac12{\pmb w}Q{\pmb w}^T
+
\tilde{\pmb c}{\pmb w}^T
+
\tilde{k}
-
\lambda
\left(
\tilde{\pmb\mu}{\pmb w}^T-\tilde{\mu}_0
\right).
\end{align*}

Computing the partial derivatives of this function with respect to $w_1,\,w_2,\,\ldots,\,w_n$ and $\lambda$ yields
\begin{align}\label{proof:P5:system}
\begin{cases}
Q{\pmb w}^T+\tilde{\pmb c}^{\,T}
-\lambda\tilde{\pmb\mu}^{\,T}
={\pmb0}^T\\
\tilde{\pmb\mu}{\pmb w}^T=\tilde{\mu}_0
\end{cases}.
\end{align}

Multiplying the first equation of \eqref{proof:P5:system} by $Q^{-1}$ gives
\begin{align}\label{proof:P5:w1}
{\pmb w}^T
=
Q^{-1}
\left(
\lambda\tilde{\pmb\mu}^{\,T}
-
\tilde{\pmb c}^{\,T}
\right).
\end{align}

Substituting \eqref{proof:P5:w1} into the second equation of
\eqref{proof:P5:system} yields
\begin{align*}
\tilde{\pmb\mu}Q^{-1}
\left(
\lambda\tilde{\pmb\mu}^{\,T}
-
\tilde{\pmb c}^{\,T}
\right)
=
\tilde{\mu}_0,
\end{align*}

and consequently,
\begin{align*}
\lambda
=
\frac{
\tilde{\mu}_0
+
\tilde{\pmb\mu}Q^{-1}\tilde{\pmb c}^{\,T}
}
{
\tilde{\pmb\mu}Q^{-1}\tilde{\pmb\mu}^{\,T}
}.
\end{align*}

Substituting this expression into \eqref{proof:P5:w1}, we obtain
\begin{align}\label{proof:P5:w2}
{\pmb w}^T
=
Q^{-1}
\left(
\frac{
\tilde{\mu}_0
+
\tilde{\pmb\mu}Q^{-1}\tilde{\pmb c}^{\,T}
}
{
\tilde{\pmb\mu}Q^{-1}\tilde{\pmb\mu}^{\,T}
}
\tilde{\pmb\mu}^{\,T}
-
\tilde{\pmb c}^{\,T}
\right).
\end{align}

Since the generalized quadratic risk measure is strictly convex, every stationary point is the unique global minimizer. Therefore, \eqref{proof:P5:w2} is the unique minimizer.

Substituting \eqref{proof:P5:w2} into the generalized quadratic risk measure gives
\begin{align*}
R_P({\pmb w})
=
\frac12{\pmb w}Q{\pmb w}^T
+
\tilde{\pmb c}{\pmb w}^T
+
\tilde{k}
=
\frac{
\left(
\tilde{\mu}_0
+
\tilde{\pmb\mu}Q^{-1}\tilde{\pmb c}^{\,T}
\right)^2
}
{
2\tilde{\pmb\mu}Q^{-1}\tilde{\pmb\mu}^{\,T}
}
-
\frac12
\tilde{\pmb c}Q^{-1}\tilde{\pmb c}^{\,T}
+
\tilde{k}.
\end{align*}

Since $\mu_P=\tilde{\mu}_0+r_f$, solving the above identity for $\mu_P$ yields
\begin{align*}
\mu_P-r_f
=
\pm
\sqrt{
\left(
2R_P({\pmb w})
+
\tilde{\pmb c}Q^{-1}\tilde{\pmb c}^{\,T}
-
2\tilde{k}
\right)
\left(
\tilde{\pmb\mu}Q^{-1}\tilde{\pmb\mu}^{\,T}
\right)
}
-
\tilde{\pmb\mu}Q^{-1}\tilde{\pmb c}^{\,T}.
\end{align*}
Since $\mu_P\geqslant r_f$ on the efficient frontier, the positive branch is selected and the proof follows.
\end{proof}

\begin{proof}[Proof of Proposition \ref{P6}]
The common point of the minimal risk frontiers in Proposition \ref{P4} with $\mu\in\mathbb{R}$ and Proposition \ref{P5} corresponds to a portfolio satisfying $w_f=0$. Since $w_f=1-{\pmb1}{\pmb w}^T$, this is equivalent to imposing the condition ${\pmb1}{\pmb w}^T=1$. For such a portfolio, $c_fw_f=0$. However, $c_f$ is not necessarily zero. Therefore, after imposing $w_f=0$, Proposition~\ref{P5} reduces to the optimization problem considered in Proposition~\ref{P4} with ${\pmb c}=\tilde{\pmb c}$ and $k=\tilde{k}$. Applying the budget constraint ${\pmb 1}{\pmb w}^T=1$ to identity \eqref{proof:P5:w2} yields

\begin{align*}
\frac{
\tilde{\mu}_0
+
\tilde{\pmb\mu}Q^{-1}\tilde{\pmb c}^{T}
}
{
\tilde{\pmb\mu}Q^{-1}\tilde{\pmb\mu}^{\,T}
}=
\frac{1+{\pmb 1}Q^{-1}\tilde{\pmb c}^T}{{\pmb 1}Q^{-1}{\tilde{\pmb \mu}^T}}
\end{align*}
and, by substituting this identity back into \eqref{proof:P5:w2}, we obtain
\begin{align}\label{establishes}
{\pmb w}^T=Q^{-1}\left(\frac{1+{\pmb 1}Q^{-1}\tilde{\pmb c}^T}{{\pmb 1}Q^{-1}\tilde{\pmb \mu}^T}\tilde{\pmb \mu}^T-\tilde{\pmb c}^T\right).
\end{align}

The excess return $\tilde{\mu}_{T}=\mu_T-r_f$ of the common portfolio, denote it by $T$, in \eqref{establishes} is 
\begin{align*}
\tilde{\mu}_{T}=\tilde{\pmb \mu}Q^{-1}\left(\frac{1+{\pmb 1}Q^{-1}\tilde{\pmb c}^T}{{\pmb 1}Q^{-1}\tilde{\pmb \mu}^T}\tilde{\pmb \mu}^T-\tilde{\pmb c}^T\right).
\end{align*}
On the other hand, by \eqref{cmc}
\begin{align*}
\tilde{\mu}_{R \min}=\mu_{R \min}-r_f=-\tilde{\pmb \mu} Q^{-1}\tilde{\pmb c}^T.
\end{align*}
Therefore, the common point of the minimal risk frontiers is also a common point of the efficient frontiers if
\begin{align*}
\tilde{\mu}_{T}-\tilde{\mu}_{R \min}=\frac{(1+{\pmb 1}Q^{-1}\tilde{\pmb c}^T)\tilde{\pmb \mu}Q^{-1}\tilde{\pmb \mu}^T}{{\pmb 1}Q^{-1}\tilde{\pmb \mu}^T}>0.
\end{align*}

To prove tangency, it suffices to observe $H(\mu)=F_4(\mu)- F_5(\mu)\geqslant0$ for all $\mu\in\mathbb{R}$, where
\begin{align*}
F_4(\mu)
&=
\min
\left\{
R_P({\pmb w}):
{\pmb1}{\pmb w}^T=1,\,
{\pmb\mu}{\pmb w}^T=\mu
\right\},\\
F_5(\mu)
&=
\min
\left\{
R_P({\pmb w},w_f):
{\pmb1}{\pmb w}^T+w_f=1,\,
{\pmb\mu}{\pmb w}^T+r_fw_f=\mu
\right\}.
\end{align*}
Indeed, the feasible set defining $F_4(\mu)$ is contained in that
defining $F_5(\mu)$, since every portfolio feasible for $F_4(\mu)$ is
also feasible for $F_5(\mu)$ by taking $w_f=0$. Portfolio \eqref{establishes} is the unique efficient portfolio satisfying
$w_f=0$, and therefore the corresponding expected return
$\mu^\ast$ is unique, so $H(\mu^\ast)=0$. Hence $\mu^\ast$ is a minimum of the
differentiable function $H$, and therefore $H'(\mu^\ast)=F_4'(\mu^\ast)-F_5'(\mu^\ast)=0$,
showing that the two efficient frontiers are tangent at their unique
common point.
\end{proof}

\begin{proof}[Proof of Proposition \ref{P7}]
Let ${\pmb w}_1,\, {\pmb w}_2,\,\ldots,\,{\pmb w}_m\in\mathcal{P}$ denote the portfolios, constructed from the same $n$ risky assets $A_1,\, A_2,\, \ldots,\, A_n$. Assume that $\mu_{0,1},\,\mu_{0,2},\,\ldots,\,\mu_{0,m}$ are the expected returns and ${\pmb c}_1,\,{\pmb c}_2,\,\ldots,\,{\pmb c}_m$ are linear adjustment vectors of these portfolios respectively. If $(\lambda_1,\,\lambda_2,\,\ldots,\,\lambda_m)\in\mathbb{R}^m$, $\mu_0\in\mathbb{R}$ satisfies
\begin{align*}
\begin{cases}
\lambda_1+\lambda_2+\cdots+\lambda_m=1,\\
\lambda_1\mu_{0,\,1}
+\lambda_2\mu_{0,\,2}
+\cdots
+\lambda_m\mu_{0,\,m}
=
\mu_0
,
\end{cases}
\end{align*}
then, by \eqref{P3:w}, we obtain
\begin{align*}
&\lambda_1{\pmb w}_1^T+\lambda_2{\pmb w}_2^T+\cdots+\lambda_m{\pmb w}_m^T\\
&=\begin{pmatrix}
        Q^{-1}{\pmb \mu}^T &Q^{-1}{\pmb 1}^T
    \end{pmatrix} 
    \begin{pmatrix}
    {\pmb \mu}Q^{-1}{\pmb \mu}^T & {\pmb \mu}Q^{-1}{\pmb 1}^T \\
    {\pmb \mu}Q^{-1}{\pmb 1}^T & {\pmb 1}Q^{-1}{\pmb 1}^T
\end{pmatrix}^{-1}
\begin{pmatrix}
    \lambda_1(\mu_{0,1}+{\pmb \mu}Q^{-1}{\pmb c}_1^T) \\
    \lambda_1(1+{\pmb 1}Q^{-1}{\pmb c}_1^T)
\end{pmatrix} -\lambda_1Q^{-1}{\pmb c}_1^T \\
&+
\begin{pmatrix}
        Q^{-1}{\pmb \mu}^T &Q^{-1}{\pmb 1}^T
    \end{pmatrix} 
    \begin{pmatrix}
    {\pmb \mu}Q^{-1}{\pmb \mu}^T & {\pmb \mu}Q^{-1}{\pmb 1}^T \\
    {\pmb \mu}Q^{-1}{\pmb 1}^T & {\pmb 1}Q^{-1}{\pmb 1}^T
\end{pmatrix}^{-1}
\begin{pmatrix}
    \lambda_2(\mu_{0,2}+{\pmb \mu}Q^{-1}{\pmb c}_2^T) \\
    \lambda_2(1+{\pmb 1}Q^{-1}{\pmb c}_2^T)
\end{pmatrix} -\lambda_2Q^{-1}{\pmb c}_2^T\\
&\,\,\,\vdots\\
&+
\begin{pmatrix}
        Q^{-1}{\pmb \mu}^T &Q^{-1}{\pmb 1}^T
    \end{pmatrix} 
    \begin{pmatrix}
    {\pmb \mu}Q^{-1}{\pmb \mu}^T & {\pmb \mu}Q^{-1}{\pmb 1}^T \\
    {\pmb \mu}Q^{-1}{\pmb 1}^T & {\pmb 1}Q^{-1}{\pmb 1}^T
\end{pmatrix}^{-1}
\begin{pmatrix}
    \lambda_m(\mu_{0,m}+{\pmb \mu}Q^{-1}{\pmb c}_m^T) \\
    \lambda_m(1+{\pmb 1}Q^{-1}{\pmb c}_m^T)
\end{pmatrix} -\lambda_mQ^{-1}{\pmb c}_m^T\\
&=\begin{pmatrix}
        Q^{-1}{\pmb \mu}^T &Q^{-1}{\pmb 1}^T
    \end{pmatrix} 
    \begin{pmatrix}
    {\pmb \mu}Q^{-1}{\pmb \mu}^T & {\pmb \mu}Q^{-1}{\pmb 1}^T \\
    {\pmb \mu}Q^{-1}{\pmb 1}^T & {\pmb 1}Q^{-1}{\pmb 1}^T
\end{pmatrix}^{-1}
\begin{pmatrix}
   \mu_0+{\pmb \mu}Q^{-1}{\pmb c}^T \\
    1+{\pmb 1}Q^{-1}{\pmb c}^T
\end{pmatrix} -Q^{-1}{\pmb c}^T.
\end{align*}
Thus, the portfolio $\lambda_1{\pmb w}_1+\lambda_2{\pmb w}_2+\cdots+\lambda_m{\pmb w}_m$ with ${\pmb c}=\lambda_1{\pmb c}_1+\lambda_2{\pmb c}_2+\ldots+\lambda_m{\pmb c}_m$ also satisfies \eqref{P3:w}, and belongs to the upper branch of the hyperbola \eqref{P4:eff:frontier} (is efficient) if $\mu_0=\lambda_1\mu_{0,\,1}+\lambda_2\mu_{0,\,2}+\cdots+\lambda_m\mu_{0,m}\geqslant \mu_{\min}$, where $\mu_{\min}$ denotes the expected return of the minimum risk portfolio, computed by \eqref{P1_Min_VaR_Portf}, whose linear adjustment vector is ${\pmb c}$.
\end{proof}

\begin{proof}[Proof of Proposition \ref{P8}]
Consider the Lagrangian
\begin{align*}
    L({\pmb w},\,\lambda) = {\pmb w}{\pmb \mu}^T - \left(\frac{1}{2}{\pmb w}Q{\pmb w}^T + {\pmb c}{\pmb w}^T + k\right) -\lambda\left({\pmb 1}{\pmb w}^T - 1\right).
\end{align*}
Computing the partial derivatives with respect to
$w_1,\,w_2,\,\ldots,\,w_n$ and $\lambda$ yields
\begin{align}\label{proof:P8:syst}
\begin{cases}
{\pmb \mu}^T -Q{\pmb w}^T-{\pmb c}^T-\lambda{\pmb 1}^T={\pmb 0}^T \\
{\pmb 1}{\pmb w}^T=1
\end{cases}.
\end{align}
Multiplying the first equation of
\eqref{proof:P8:syst} by $Q^{-1}$ gives
\begin{align}\label{proof:P8:w1}
    {\pmb w}^T = Q^{-1}\left({\pmb \mu}-{\pmb c}-\lambda{\pmb 1}\right)^T.
\end{align}
Inserting this expression into the second equation of the system yields
\begin{align*}
1 = {\pmb 1}Q^{-1}\left({\pmb \mu}-{\pmb c}-\lambda{\pmb 1}\right)^T
\qquad
\implies 
\qquad
\lambda = \frac{{\pmb 1}Q^{-1}{\pmb \mu}^T - {\pmb 1}Q^{-1}{\pmb c}^T - 1}{{\pmb 1}Q^{-1}{\pmb 1}^T}
\end{align*}

Substituting the above identity into \eqref{proof:P8:w1} yields
\begin{align}\label{proof:P8:w2}
    {\pmb w}^T = \frac{1+{\pmb 1}Q^{-1}{\pmb c}^T-{\pmb 1}Q^{-1}{\pmb \mu}^T}{{\pmb 1}Q^{-1}{\pmb 1}^T}Q^{-1}{\pmb 1}^T + Q^{-1}{\pmb \mu}^T- Q^{-1}{\pmb c}^T.
\end{align}

Since $Q$ is positive definite, the Hessian of $L$
with respect to ${\pmb w}$ equals $-Q$, which is negative definite. Therefore, \eqref{proof:P8:w2} is the unique global maximizer.

Finally, substituting
\eqref{proof:P8:w2}
into the utility function and simplifying, we obtain
\begin{align*}
\max\limits_{{\pmb w}\in\mathcal{P}}U({\pmb w}) 
=\frac12\left(
c-2e+g-\frac{(b-f-1)^2}{a}
\right)-k,
\end{align*}
where $a = {\pmb 1}Q^{-1}{\pmb 1}^T$, $b={\pmb 1}Q^{-1}{\pmb \mu}^T$, $c={\pmb \mu}Q^{-1}{\pmb \mu}^T$, $e ={\pmb \mu}Q^{-1}{\pmb c}^T $, $f = {\pmb 1}Q^{-1}{\pmb c}^T$, and $g = {\pmb c}Q^{-1}{\pmb c}^T$.

The efficiency of the portfolio given by \eqref{proof:P8:w2} follows by
contradiction. Suppose that it is not efficient. Then there exists another
portfolio ${\pmb w^\ast}$ satisfying ${\pmb1}{\pmb w^\ast}^T=1$ such that
\[
{\pmb\mu}{\pmb w^\ast}^T
\geqslant
{\pmb\mu}{\pmb w}^T,
\qquad
R_P({\pmb w^\ast})
\leqslant
R_P({\pmb w}),
\]
where at least one of the inequalities is strict. Consequently,
\[
{\pmb\mu}{\pmb w^\ast}^T
-
R_P({\pmb w^\ast})
>
{\pmb\mu}{\pmb w}^T
-
R_P({\pmb w}),
\]
which contradicts the fact that ${\pmb w}$ maximizes the utility function.
\end{proof}

\begin{proof}[Proof of Proposition \ref{P9}]
Consider the Lagrangian
\begin{align*}
    L({\pmb w}_F,\,\lambda) = r_fw_f + {\pmb w}{\pmb \mu}^T - \left(\frac{1}{2}{\pmb w}Q{\pmb w}^T + {\pmb c}{\pmb w}^T + c_fw_f + k\right) - \lambda({\pmb 1}{\pmb w}^T +w_f - 1).
\end{align*}
By computing the partial derivatives of $L({\pmb w}_F,\,\lambda)$ with respect to
$w_1,\,w_2,\,\ldots,\,w_n$, $w_f$ and $\lambda$ we obtain
\begin{align}\label{proof:P9:syst}
\begin{cases}
{\pmb \mu}^T -Q{\pmb w}^T-{\pmb c}^T-\lambda{\pmb 1}^T={\pmb 0}^T \\
r_f -c_f - \lambda = 0
\\
{\pmb 1}{\pmb w}^T + w_f=1
\end{cases}.
\end{align}

The second equation of the system \eqref{proof:P9:syst} yields $\lambda = r_f - c_f$. Multiplying the first equation of
\eqref{proof:P9:syst} by $Q^{-1}$ and inserting $\lambda$ value gives
\begin{align}\label{proof:P9:w1}
{\pmb w}^T = Q^{-1}\left(({\pmb \mu}-r_f{\pmb 1})^T - ({\pmb c}-c_f{\pmb 1})^T\right) = Q^{-1}\left({\tilde{\pmb \mu}}^T - {\tilde{\pmb c}}^T\right).
\end{align}
Since $Q$ is positive definite, the Hessian of $L$
with respect to ${\pmb w}$ equals $-Q$, which is negative definite. Therefore, \eqref{proof:P9:w1} is the unique global maximizer.

Inserting \eqref{proof:P9:w1} into the third equation of the system \eqref{proof:P9:syst} yields
\begin{align*}
w_f = 1-{\pmb 1}{\pmb w}^T=1 - {\pmb 1}Q^{-1}(\tilde{\pmb \mu} - \tilde{\pmb c})^T.
\end{align*}

By using
\begin{align*}
U({\pmb w}_F)
=
r_f
+
{\pmb w}({\pmb\mu}-r_f{\pmb1})^T
-
\left(
\frac12{\pmb w}Q{\pmb w}^T
+
({\pmb c}-c_f{\pmb1}){\pmb w}^T
+
(k+c_f)
\right),
\end{align*}
we obtain
\begin{align*}
\max\limits_{{\pmb w}\in\mathcal P}U({\pmb w})
&=
r_f-\tilde{k}
+
(\tilde{\pmb\mu}-\tilde{\pmb c})
Q^{-1}
(\tilde{\pmb\mu}-\tilde{\pmb c})^T
-\frac12
(\tilde{\pmb\mu}-\tilde{\pmb c})
Q^{-1}
(\tilde{\pmb\mu}-\tilde{\pmb c})^T
\\
&=
r_f-\tilde{k}
+\frac12
(\tilde{\pmb\mu}-\tilde{\pmb c})
Q^{-1}
(\tilde{\pmb\mu}-\tilde{\pmb c})^T
\\
&=
r_f-\tilde{k}+\frac{p+q-2s}{2},
\end{align*}
where
\begin{align*}
p=\tilde{\pmb\mu}Q^{-1}\tilde{\pmb\mu}^T,
\qquad
q=\tilde{\pmb c}Q^{-1}\tilde{\pmb c}^T,
\qquad
s=\tilde{\pmb\mu}Q^{-1}\tilde{\pmb c}^T.
\end{align*}

The efficiency of the portfolio \eqref{proof:P9:w1} follows by
the same arguments as at the end of 
\end{proof}

\section{Example}\label{sec:examples}

In this section, we present a numerical example illustrating the theoretical results formulated in Propositions \ref{P1}--\ref{P9}. The example is based on synthetic data chosen solely for illustrative purposes rather than on empirical financial data. Its purpose is to demonstrate the application of the proposed formulas and to verify the consistency of the theoretical developments. We use MATLAB \cite{MATLAB} for the computations and visualizations described. 

\begin{ex}
Consider a portfolio consisting of four hypothetical risky assets with an expected return vector ${\pmb \mu}$, symmetric positive definite matrix $Q$, and linear adjustment vector ${\pmb c}$:
\begin{align*}
{\pmb \mu} = (2,\, 5,\, 9,\, 14), \qquad 
Q = \begin{pmatrix}
6 & 2 & -1 & 0 \\
2 & 5 & 1 & 1 \\
-1 & 1 & 4 & 1 \\
0 & 1 & 1 & 3
\end{pmatrix},
\qquad
{\pmb c} = (3,\,-4,\,2,\,-1).
\end{align*}
In addition, assume that a risk-free asset is available with $r_f = 1$ and $c_f =0$, and say that short selling is permitted. We:
\begin{enumerate}[label=\textup{(\alph*)}]
\item compute the quantities appearing in
Propositions~\ref{P1}--\ref{P9}, illustrate the efficient frontiers
from Propositions~\ref{P4} and~\ref{P5}, and determine the portfolios
characterized by the remaining propositions together with their
corresponding risks and expected returns;

\item illustrate how the efficient frontier changes when the linear adjustment vector is replaced by $\lambda{\pmb c}$, where
$\lambda\in\{-5/8,\,0,\,5/8,\,5/4,\,15/8,\,5/2,\,25/8\}$;

\item illustrate how the efficient frontier changes when varying the parameter\\
$c_f\in\{-1,\,0,\,1,\,2,\,3,\,4\}$;

\item illustrate how the efficient frontier changes when varying the parameter $k=1/2\,{\pmb c}Q^{-1}{\pmb c}^T+\varepsilon$, where $\varepsilon\in\{0,\,1,\,2,\,3,\,4\}$.

\item take one more linear adjustment vector ${\pmb c}_2=-{\pmb c}$ and illustrate the generalized fund separation statement given in Proposition \ref{P7}. 
\end{enumerate}
\end{ex}
The subsequent numerical values are generally rounded to two decimal places. Consequently, minor discrepancies may occur in the reported portfolio weights (they may not sum exactly to 1), risks, returns, and other quantities.

Consider (a). For the given data, ${\pmb c}Q^{-1}{\pmb c}^T/2= 6.10$. We choose 
\[
k=\frac12{\pmb c}Q^{-1}{\pmb c}^T+1=7.10
\]
to satisfy the assumption
\[
k\geqslant\frac12{\pmb c}Q^{-1}{\pmb c}^T,
\]
which yields $R_P(\pmb{w})>0$.

Table~\ref{T1:Obtimization:Rezult} summarizes the portfolios together with their expected returns and risk values obtained from Propositions~\ref{P1}, \ref{P2}, \ref{P6}, \ref{P8} and~\ref{P9}. By $F$ we denote the risk-free asset. A portfolio without $F$ is understood as $(w_1,\,w_2,\,w_3,\,w_4)$, while with $F$ as $(w_1,\,w_2,\,w_3,\,w_4,\,w_f)$. 

\begin{table}[H]
    \centering
    \begin{tabular}{lccc}
    \toprule\toprule

    \textbf{Portfolio}
    &
    \makecell{\textbf{Assigned weights} ${\pmb w}$}
    &
    \makecell{\textbf{Expected return} $\mu$}
    &
    \makecell{\textbf{Risk} $R_P$}
    \\
    \midrule

    Minimum risk
    &
    $(-0.72,\,1.54,\,-0.67,\,0.85)$
    &
    12.16
    &
    3.06
    \\

    \midrule

    Maximum SR
    &
    $(-1.26,\,1.17,\,-0.97,\,2.05)$
    &
    23.38
    &
    6.14
    \\

    \midrule

    Tangency point
    &
    $(-1.03,\,1.32,\,-0.84,\,1.55)$
    &
    18.63
    &
    4.09
    \\

    \midrule

    \makecell[l]{Maximum utility\\(without $F$)}
    &
    $(-1.70,\,0.87,\,-1.21,\,3.04)$
    &
    32.59
    &
    13.28
    \\

    \midrule

    \makecell[l]{Maximum utility\\(with $F$)}
    &
    $(-0.67,\,1.00,\,0.00,\,4.33,\,-3.67)$
    &
    60.67
    &
    31.77
    \\

    \bottomrule
    \end{tabular}
    \caption{Selected portfolios with their expected returns and corresponding risks.}
    \label{T1:Obtimization:Rezult}
\end{table}

Table \ref{T2:Obtimization:Rezult} presents the efficient frontiers derived in Proposition \ref{P4}, Corollary \ref{C}, and Proposition \ref{P5}.

\begin{table}[H]
    \centering
\begin{tabular}{ll}
\toprule\toprule
\multicolumn{2}{c}{\textbf{Efficient frontiers}}\\
\midrule

\makecell{Without risk-free asset}
&
\makecell[l]{
$R_P = 0.02\mu^2 -0.60\mu + 6.68,\quad \mu\geqslant 12.16,$\\
$\mu = \sqrt{40.85\,(R_P-3.06)}+12.16,\quad R_P\geqslant 3.06$
}
\\
\midrule

\makecell{With risk-free asset}
&
$\mu=\sqrt{61.53\,(2R_P-2.00)}-0.87,\quad R_P\geqslant 1.00$
\\
\bottomrule
\end{tabular}
    \caption{Analytic expressions of the efficient frontiers with and without a risk-free asset.}
    \label{T2:Obtimization:Rezult}
\end{table}
Figure \ref{Figurea_} depicts the portfolios and curves reported in Tables \ref{T1:Obtimization:Rezult} and \ref{T2:Obtimization:Rezult}. In addition, to illustrate the mutual fund separation given in Proposition \ref{P7}, we randomly select three efficient portfolios without a risk-free asset:
\begin{align*}
&{\pmb w}_1=(-1.13,\,1.26,\,-0.90,\,1.77),\quad \mu_{{\pmb w}_1}=20.73,\quad R_P({\pmb w}_1) = 4.86,\\
&{\pmb w}_2=(-1.74,\,0.83,\,-1.23,\,3.14),\quad \mu_{{\pmb w}_2} =33.59,\quad R_P({\pmb w}_2) = 14.30,\\
&{\pmb w}_3=(-2.36,\,0.41,\,-1.57,\,4.52),\quad \mu_{{\pmb w}_3} =46.43,\quad R_P({\pmb w}_3) = 31.82.
\end{align*}
And three efficient portfolios that include the risk-free asset:
\begin{align*}
{\pmb w}_1^*
    &=(-1.16,\,1.45,\,-1.15,\,0.50,\,1.37),
&
\mu_{{\pmb w}_1^*}
    &=2.93,
&
R_P({\pmb w}_1^*)
    &=1.12,
\\
{\pmb w}_2^*
    &=(-1.10,\,1.39,\,-1.00,\,1.01,\,0.71),
&
\mu_{{\pmb w}_2^*}
    &=10.54,
&
R_P({\pmb w}_2^*)
    &=2.06,
\\
{\pmb w}_3^*
    &=(-0.91,\,1.22,\,-0.58,\,2.42,\,-1.15),
&
\mu_{{\pmb w}_3^*}
    &=31.83,
&
R_P({\pmb w}_3^*)
    &=9.69.
\end{align*}
Then by selecting the target returns $\mu_0 = 51.57$, $\mu_0^* = 34.72$ on the corresponding efficient frontiers and choosing the coefficients $\lambda_1 =-0.37,\,\lambda_2=0.33,\,\lambda_3=1.03$ and $\lambda^*_1=-0.20,\,\lambda^*_2=0.13,\,\lambda^*_3=1.06$ that satisfy the system \eqref{syst}, we conclude that the linear combinations $\lambda_1{\pmb w}_1+\lambda_2\pmb{w}_2+\lambda_3{\pmb w}_3$ and $\lambda^*_1{\pmb w}^*_1+\lambda^*_2\pmb{w}^*_2+\lambda^*_3{\pmb w}^*_3$ are also the efficient portfolios; see Figure \ref{Figurea_}. 

\begin{figure}[H]
\centering
\includegraphics[width=160mm,height=125mm]{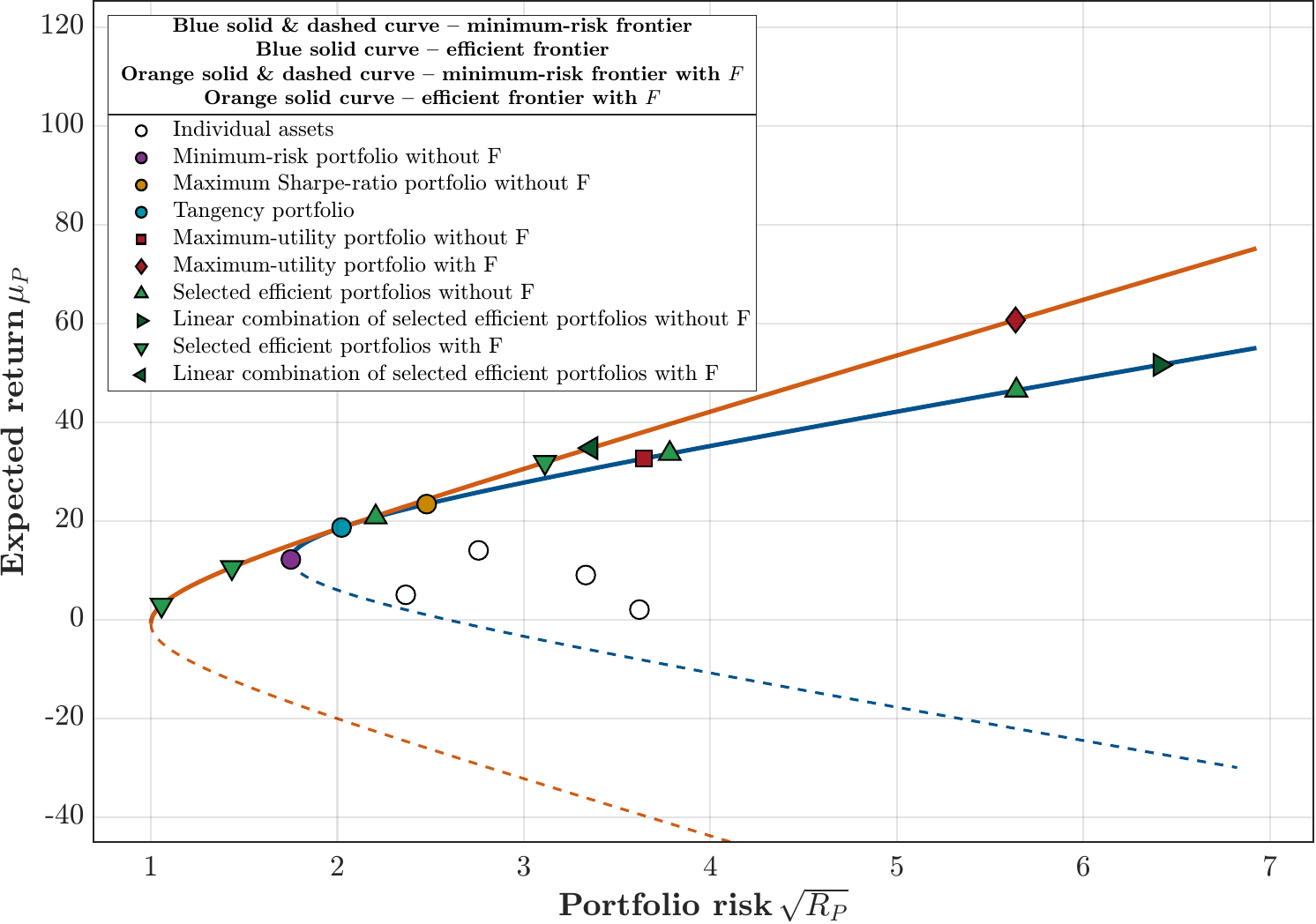}
\caption{Visualization of portfolio selection theory with a general quadratic risk measure}
\label{Figurea_}
\end{figure}

Consider (b). For each $\lambda\in\{-5/8,\,0,\,5/8,\,5/4,\,15/8,\,5/2,\,25/8\}$, we set up the risk measures
\begin{align}\label{k_no_shift}
R_{P,\,\lambda}({\pmb w})=\frac{1}{2}{\pmb w}Q{\pmb w}^T+\lambda {\pmb c}{\pmb w}^T+k,\qquad k=\frac{1}{2}\lambda^2{\pmb c}Q^{-1}{\pmb c}^T, 
\end{align}
and depict the minimal risk curves as provided in Proposition \ref{P4} and Proposition \ref{P5} with appropriate replacement of ${\pmb w}$ with ${\pmb w}_F$. Obviously, $\lambda=0$ corresponds to the classical mean-variance case. Moreover, if $k$ is as chosen, without additional increment, then the efficient frontier in \eqref{cmc} becomes a line, see the right-hand-side picture in Figure \ref{Figureb}.

\begin{figure}[H]
\centering
\includegraphics[width=160mm,height=80mm]{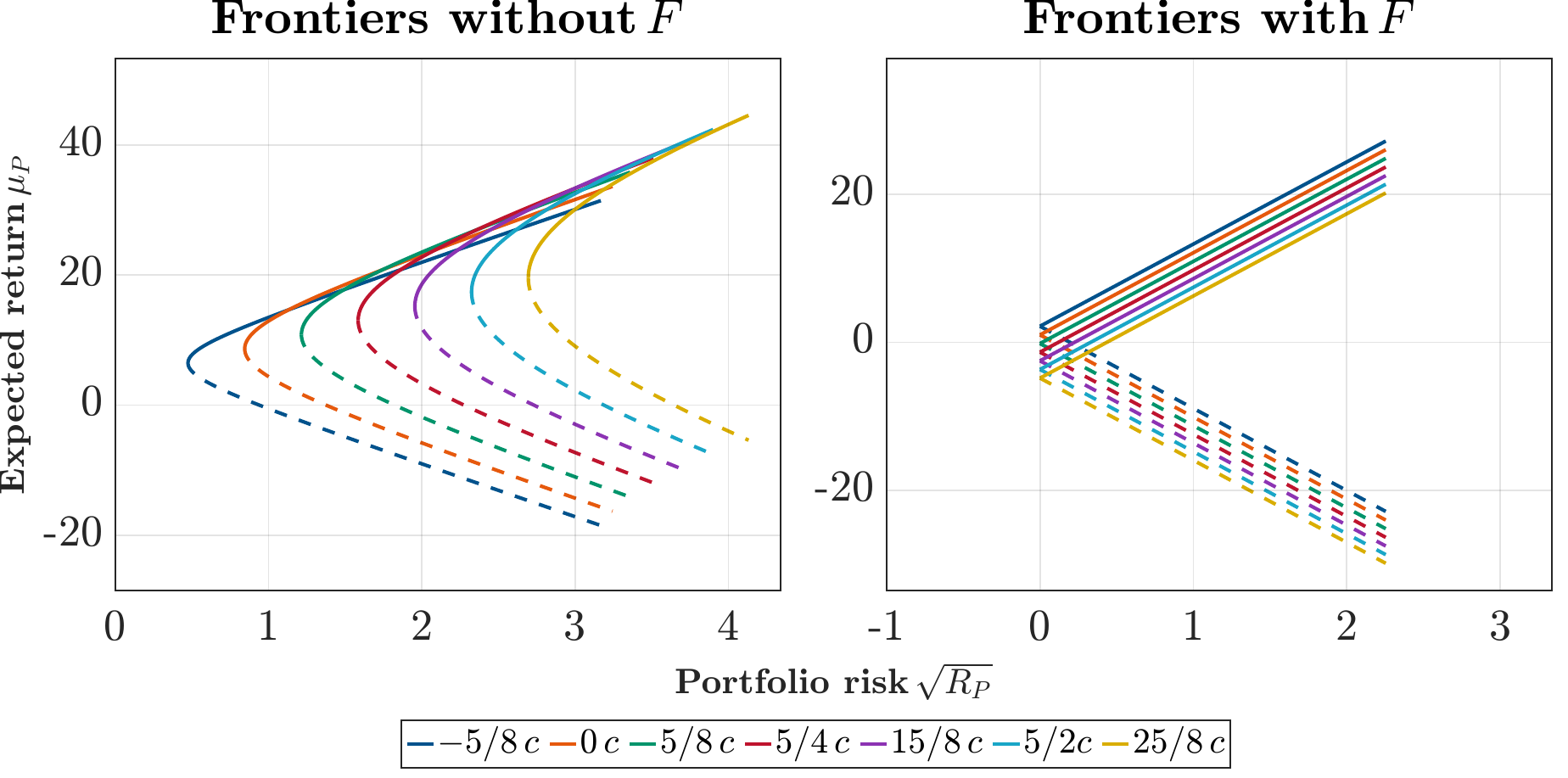}
\caption{Effect of $\lambda{\pmb c}$ for the efficient frontiers}
\label{Figureb}
\end{figure}

If we replace $k$ in \eqref{k_no_shift} with $k=\lambda^2{\pmb c}Q^{-1}{\pmb c}/2+1$ we then obtain a different efficient frontiers compared to those in Figure \ref{Figureb}, see Figure \ref{Figureb2}. 

\begin{figure}[H]
\centering
\includegraphics[width=160mm,height=80mm]{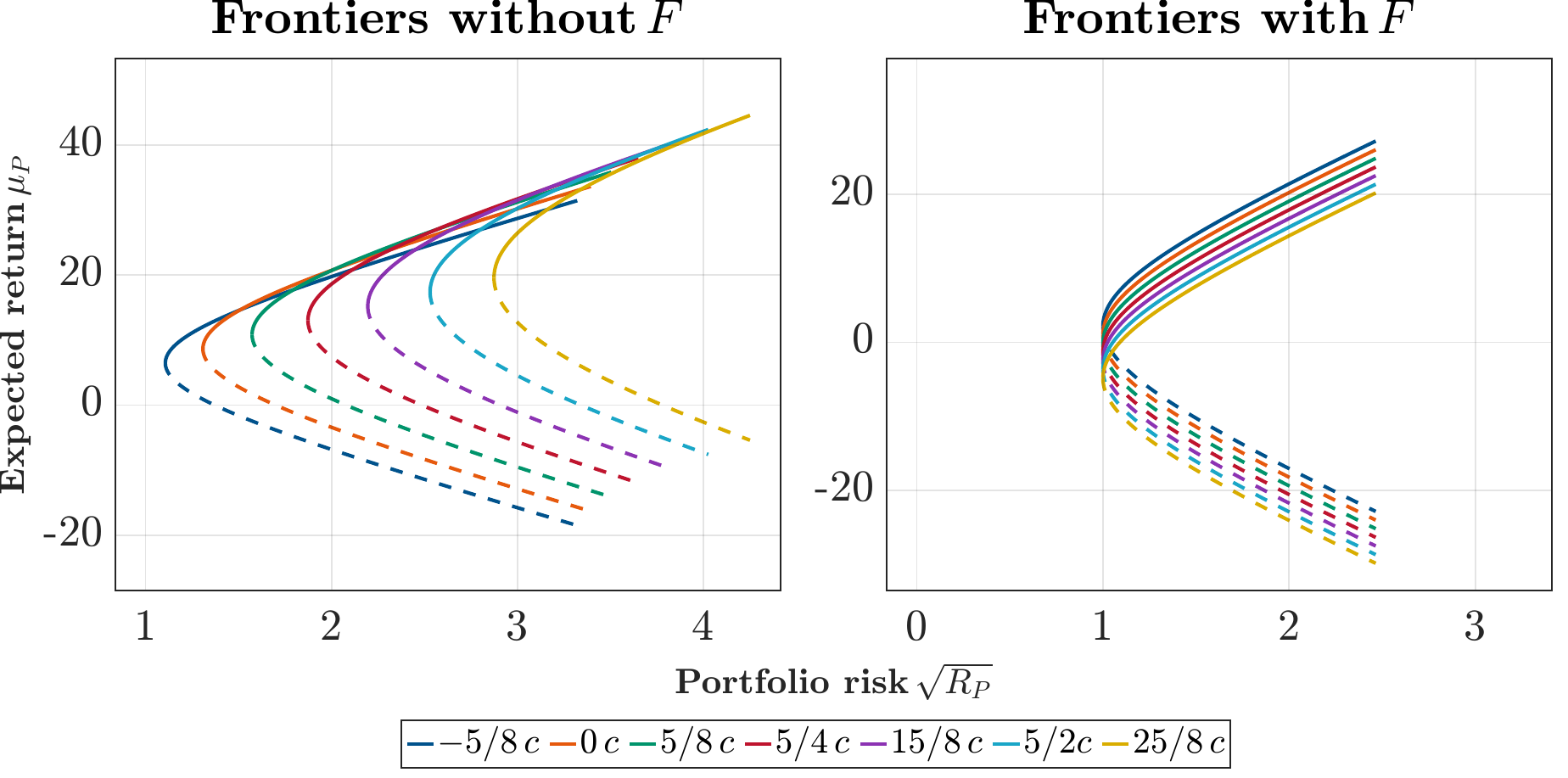}
\caption{Effect of $\lambda{\pmb c}$ for the efficient frontiers with adjusted constant $k$}
\label{Figureb2}
\end{figure}

Figures \ref{Figureb} and \ref{Figureb2} show that increasing $\lambda$ shifts the efficient risk–return frontier toward higher risk and return without $F$, but toward lower return under the same risk when $F$ is available.

\medskip

Consider (c). For each $c_f\in\{-1/2,\,0,\,1,\,2,\,3,\,4\}$, we set up the risk measures
\begin{align*}
R_{P,\,c_f}({\pmb w})=\frac{1}{2}{\pmb w}Q{\pmb w}^T+\tilde{\pmb c}{\pmb w}^T+\tilde{k},\qquad {\tilde k}= \frac{1}{2}{\pmb c}Q^{-1}{\pmb c}^T + 1 + c_f,\qquad \tilde{\pmb c}={\pmb c}-{\pmb 1}c_f,
\end{align*}
and depict the efficient frontiers with the risk-free asset given by
Proposition~\ref{P5}. Obviously, $c_f=0$ corresponds to case (a)
with the risk-free investment appended: this is the orange curve in Figure \ref{Figurec_} with $c_f=0$ and also the orange curve in Figure \ref{Figurea_}. The blue curve in all panels of Figure~\ref{Figurec_} is fixed and represents the efficient frontier without the risk-free asset for the parameters
$Q$, ${\pmb c}$, and $k$ specified in case (a). In other words, the blue curve in all the panels of Figure \ref{Figurec_} is the same as the blue curve in Figure \ref{Figurea_}.
\begin{figure}[H]
\centering
\includegraphics[width=150mm,height=175mm]{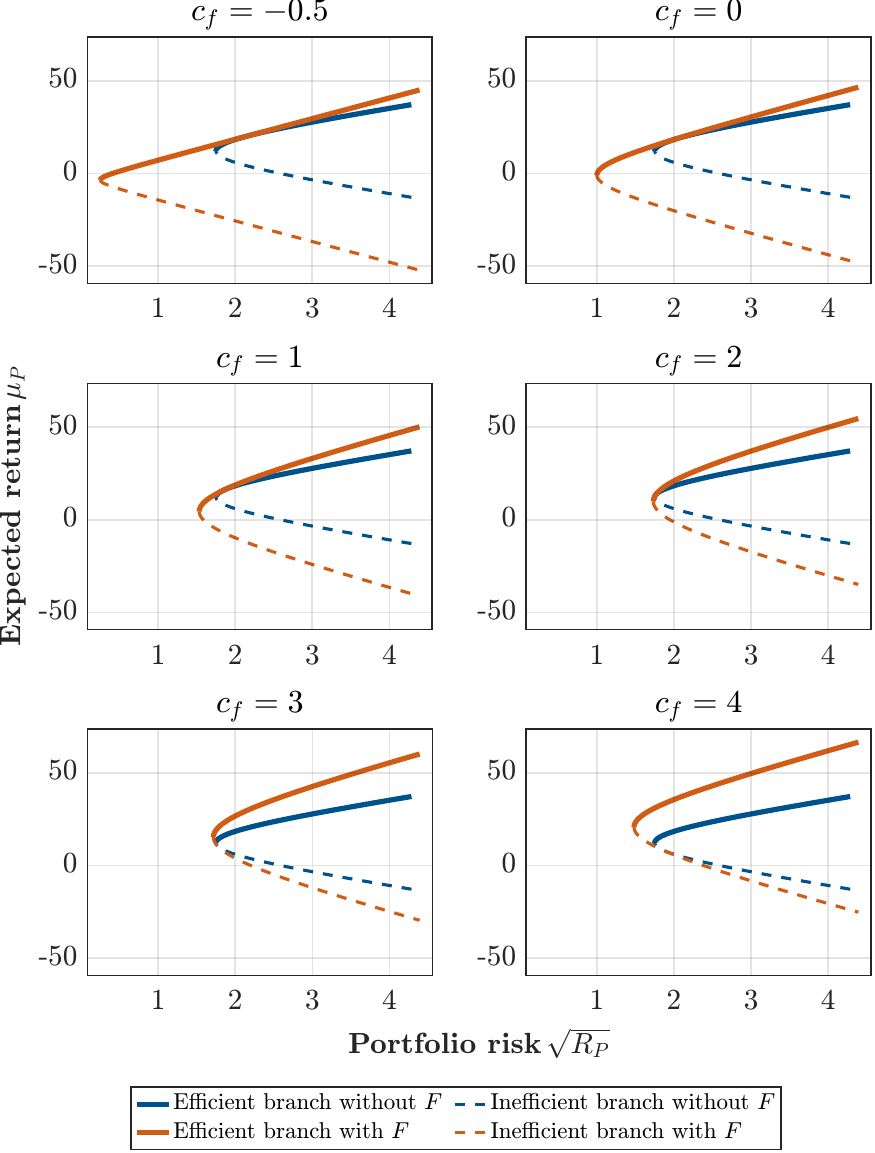}
\caption{Effect of $c_f$ on the efficient frontier with a risk-free asset}
\label{Figurec_}
\end{figure}
Figure \ref{Figurec_} shows that increasing $c_f$, in general, moves the minimum-risk point of the efficient frontier with the risk-free asset toward higher risk and
expected return. Large positive $c_f$ ($=3,\,4$) can make the risk-free asset more attractive than risky investment as the relevant tangency point moves to the lower branch. This corresponds to the negative slope of the capital allocation line in the classical Markowitz theory.

Consider (d). For each $\varepsilon\in\{0,\,1,\,2,\,3,\,4\}$, we set up the risk measures
\begin{align*}
R_{P,\,\varepsilon}({\pmb w})=\frac{1}{2}{\pmb w}Q{\pmb w}^T+{\pmb c}{\pmb w}^T+k,\qquad  k=\frac{1}{2}{\pmb c}Q^{-1}{\pmb c}^T+\varepsilon,
\end{align*}
and depict the minimal risk curves as provided in Proposition \ref{P4} and Proposition \ref{P5} under the replacement of ${\pmb w}$ with ${\pmb w}_F$, where $r_f=1$ and $c_f=0$. Obviously, $\varepsilon=1$ corresponds to case (a). 

\begin{figure}[H]
\centering
\includegraphics[width=160mm,height=85mm]{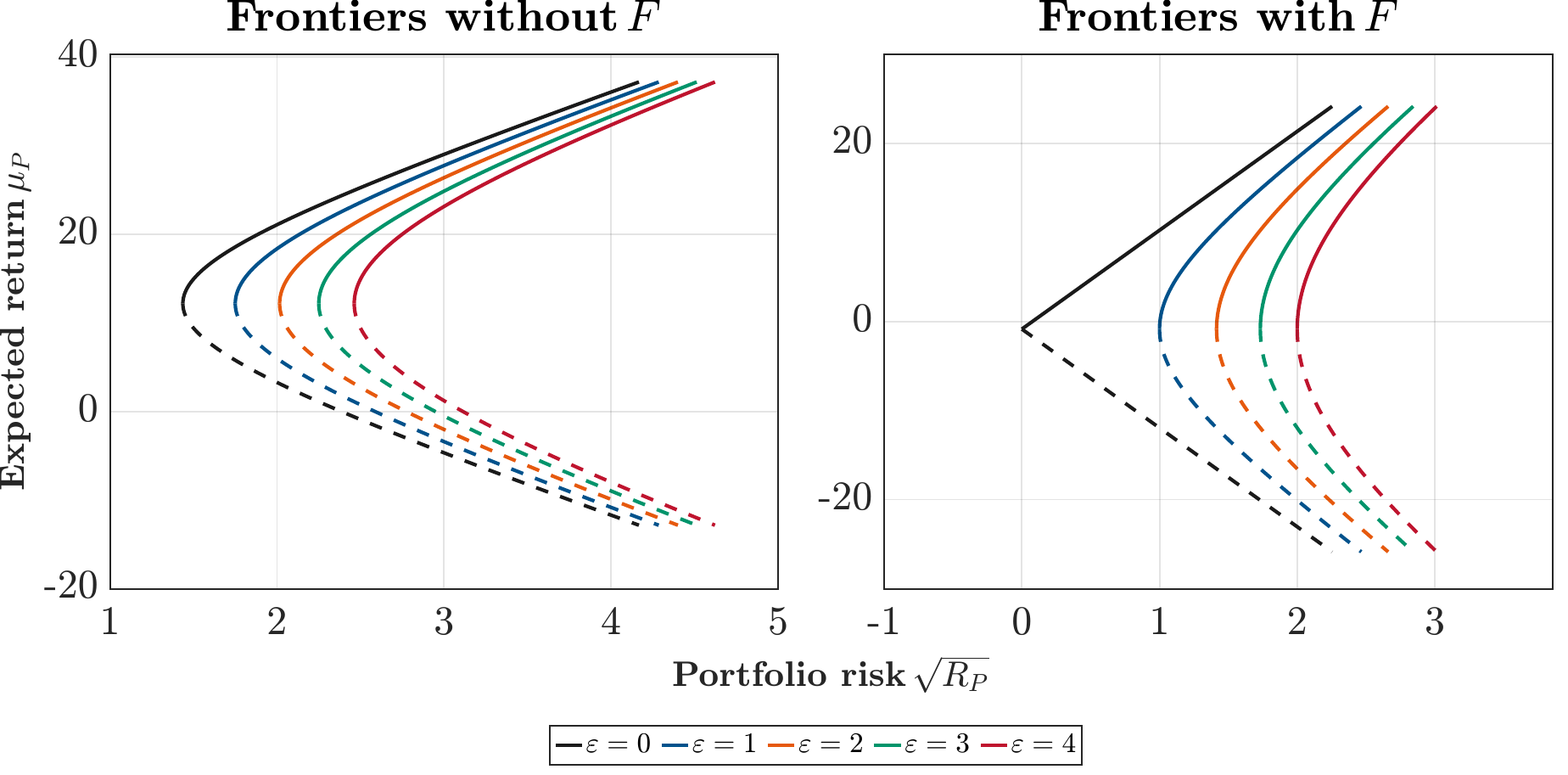}
\caption{Effect on $k$ shift for the efficient frontiers with and without risk-free investment}
\label{Figured_}
\end{figure}

Figure \ref{Figured_} illustrates that varying $k$, shifts the efficient frontiers
horizontally, with the minimum attainable risk changing according to
$\varepsilon$.

Consider (e). For the given ${\pmb c}={\pmb c}_1=(3,\,-4,\,2,\,-1)$, we take one more linear adjustment vector ${\pmb c}_2=-{\pmb c}_1$ and consider two risk measures:
\begin{align*}
R_{P,\,i}({\pmb w})=\frac{1}{2}{\pmb w}Q{\pmb w}^T+{\pmb c}_i{\pmb w}^T+k_i,\qquad  k_i=\frac{1}{2}{\pmb c}_iQ^{-1}{\pmb c}_i^T+1,\qquad i=1,\,2,
\end{align*}
with the intention to illustrate Proposition \ref{P7} more broadly than the the case (a). 

Let ${\pmb w}_{\min,\,1}$ and ${\pmb w}_{\min,\,2}$ denote the corresponding minimum-risk portfolios for the chosen linear adjustment vectors ${\pmb c}_1$ and ${\pmb c}_2$. Then, by Proposition \ref{P1}, these portfolios with their underlying returns and risks are (the blue and orange circles in Figure \ref{Figuree1}):
\begin{align*}
{\pmb w}_{\min,\,1}
    &= (-0.72,\,1.54,\,-0.67,\,0.85),
&
\mu_{\min,\,1}
    &= 12.16,
&
R_P({\pmb w}_{\min,\,1})
    &= 3.06,
\\
{\pmb w}_{\min,\,2}
    &= (1.28,\,-1.46,\,1.33,\,-0.15),
&
\mu_{\min,\,2}
    &= 5.16,
&
R_P({\pmb w}_{\min,\,2})
    &= 1.06.
\end{align*}
By Proposition \ref{P4}, depict two efficient frontiers corresponding to the risk measures $R_{P,\,1}$ and $R_{P,\,2}$: the blue and orange curves in Figure \ref{Figuree1}. We then randomly select two portfolios from each of the depicted efficient frontiers (the blue and orange triangulars in Figure \ref{Figuree1}):
\begin{align*}
{\pmb w}_{1}
    &= (-0.87,\,1.43,\,-0.75,\,1.19),
&
\mu_{0,\,1}
    &= 15.27,
&
R_P({\pmb w}_{1})
    &= 3.30,
\\
{\pmb w}_{2}
    &= (0.26,\,-2.16,\,0.77,\,2.14),
&
\mu_{0,\,2}
    &= 26.51,
&
R_P({\pmb w}_{2})
    &= 12.23.
\end{align*}

Suppose that we want to form a portfolio ${\pmb w}_{\textup{new}}=\lambda_1{\pmb w}_1+\lambda_2{\pmb w}_2$, whose expected return is $\mu_0=17.59$. Then, we set up and solve the system \eqref{syst}
\begin{align}\label{syst_2x2_ex}
\begin{cases}
\lambda_1+\lambda_2=1\\
\lambda_1 \mu_{0,\,1}+\lambda_2 \mu_{0,\,2}=\mu_0
\end{cases}\qquad \Rightarrow \qquad \lambda_1=0.79,\qquad\lambda_2=0.21
\end{align}

According to Proposition \ref{P7}, the portfolio 
\begin{align*}
{\pmb w}_{\textup{new}}=\lambda_1{\pmb w}_1+\lambda_2{\pmb w}_2
=(-0.63,\,0.69,\,-0.44,\,1.38)
\end{align*}
with the corresponding linear adjustment vector 
\begin{align*}
{\pmb c}_{\textup{new}}=\lambda_1{\pmb c}_1+\lambda_2{\pmb c}_2=(1.76,\,-2.35,\,1.17,\,-0.59)
\end{align*}
is efficient because $\mu_0=17.59>10.71=\mu_{\min}$, where $\mu_{\min}$ is the return of the minimum risk portfolio with respect to the linear adjustment vector ${\pmb c}_{\textup{new}}$, see Figure \ref{Figuree1}.

\begin{figure}[H]
\centering
\includegraphics[width=160mm,height=135mm]{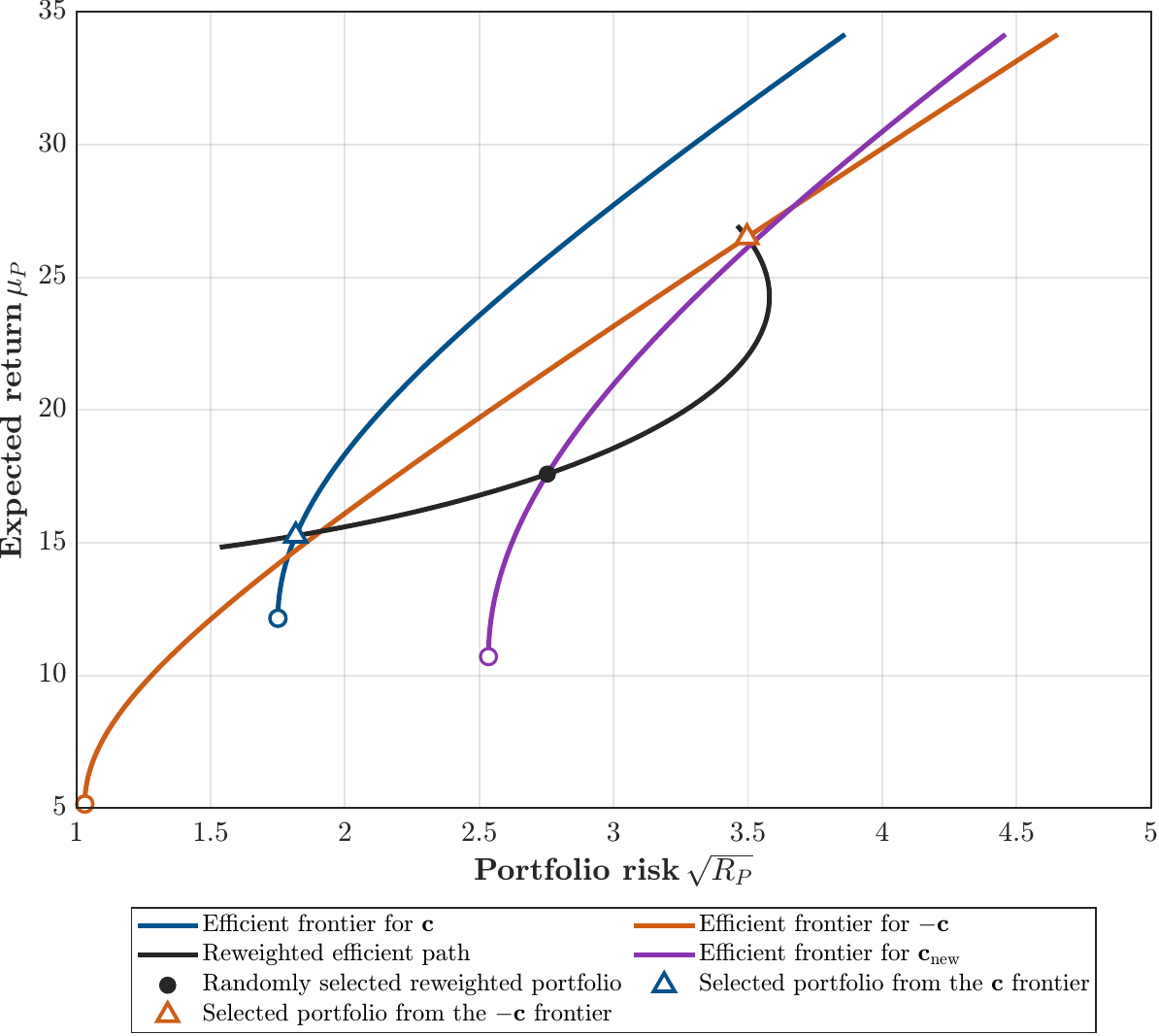}
\caption{Visual of the generalized mutual fund separation theorem}
\label{Figuree1}
\end{figure}

Let us generalize this example a bit by letting $ 14.83<\mu_0<26.95$ in the system \eqref{syst_2x2_ex}. We then obtain a set of portfolios 
\begin{align*}
\lambda {\pmb{w}_1} + (1-\lambda){\pmb{w}_2},
\end{align*}
where
\begin{align*}
-0.04=\frac{26.95-\mu_{0,\,2}}{\mu_{0,\,1}-\mu_{0,\,2}}\leqslant \lambda
\leqslant
\frac{14.83-\mu_{0,\,2}}{\mu_{0,\,1}-\mu_{0,\,2}}=1.04.
\end{align*}
The corresponding risks and returns of $\lambda {\pmb{w}_1} + (1-\lambda){\pmb{w}_2}$ are illustrated by the black curve in Figure \ref{Figuree1}. By Proposition \ref{P7}, each portfolio on the black curve in Figure \ref{Figuree1} lies on its minimal risk frontier associated with the linear adjustment vector $\lambda {\pmb c}_1+(1-\lambda){\pmb c}_2$, see the black dot and violet curve in Figure \ref{Figuree1} for the particular case when $\mu_0=17.59$.

\bibliographystyle{plain}
\bibliography{bibliography}

\end{document}